\documentclass[runningheads]{llncs}
\usepackage{amssymb}
\usepackage{graphicx}
\usepackage{color,xcolor}
\usepackage{amsmath,mathrsfs}
\usepackage{hyperref}

\usepackage{url}
\urldef{\mailsa}\path|{alfred.hofmann, ursula.barth, ingrid.haas, frank.holzwarth,|
\urldef{\mailsb}\path|anna.kramer, leonie.kunz, christine.reiss, nicole.sator,|
\urldef{\mailsc}\path|erika.siebert-cole, peter.strasser, lncs}@springer.com|

\makeatletter

\newcommand{\Rmnum}[1]{\expandafter\@slowromancap\romannumeral #1@}
\makeatother

\newcommand{\F}{\mathbb{F}}

\newcommand {\C}{{\mathcal{C}}}
\newcommand {\ccc}{{\mathbf{c}}}

\DeclareMathOperator{\wt}{wt_H}

\newtheorem{open}[theorem]{Open Problem}

\begin{document}

\title{Cyclic codes from low differentially uniform functions}


\titlerunning{Cyclic codes from functions with low differential uniformity}

\author{Sihem Mesnager \inst{1}
\and Minjia Shi\ \inst{2}
\and Hongwei Zhu\inst{2}
}
\authorrunning{S. Mesnager, M. Shi, and H. Zhu}

\institute{
Department of Mathematics, University of Paris VIII, F-93526 Saint-Denis, University Sorbonne Paris Cit\'e, LAGA, UMR 7539, CNRS, 93430 Villetaneuse and Telecom Paris, Polytechnic Institute of Paris, 91120 Palaiseau, France.\\
\email{smesnager@univ-paris8.fr}
\and School of Mathematical Sciences, Anhui University, 230601, Hefei, China.\\
 \email{smjwcl.good@163.com, zhwgood66@163.com}
}

\maketitle

\date{\today}

\begin{abstract}
 Cyclic codes have many applications in consumer electronics, communication and data storage systems due to their efficient encoding and decoding algorithms. An efficient approach to constructing cyclic codes is the sequence approach. In their articles [Discrete Math. 321, 2014] and [SIAM J. Discrete Math. 27(4), 2013],
 Ding and Zhou constructed several classes of
cyclic codes from almost perfect nonlinear (APN) functions and planar functions over
finite fields and presented some open problems on cyclic codes from highly nonlinear functions. This article focuses on these exciting works by investigating new insights in this research direction. Specifically, its objective is twofold. The first is to provide a complement with some former results and present correct proofs and statements on some known ones on the cyclic codes from the APN functions. The second is studying the cyclic codes from some known functions processing low differential uniformity. Along with this article, we shall provide answers to some open problems presented in the literature. The first one concerns Open Problem 1, proposed by Ding and Zhou in Discrete Math. 321, 2014. The two others are Open Problems 5.16 and 5.25, raised by Ding in [SIAM J. Discrete Math. 27(4), 2013].
\end{abstract}


\section{Introduction}\label{introduction}
Let $q$ be a power of a prime $p$. A linear code over $\F_q$ with parameters $[n,k,d]$ is a $k$-dimensional subspace of $\F_q^n$ with minimum Hamming  distance $d$.
Let $\tau(x_0,x_1,\cdots,x_{n-1})$ denote the cyclic shift of  the vector $(x_0,x_1,\ldots,x_{n-1})$
 by the coordinates $i\mapsto i+1 {~\rm mod~} n$ (that is, the vector $(x_{n-1},x_0, \ldots,x_{n-2})$ obtained from $(x_0,x_1,\ldots,x_{n-1})$) by $\tau$).
A linear $[n,k]$ code $\C$ over $\F_q$ is called cyclic if $\ccc\in \C$ implies $\tau(\ccc)\in \C.$
 By identifying  a vector $(c_0,c_1,\ldots,c_{n-1})\in\F_q^n$ with
 $\sum_{i=0}^{n-1}c_0x^i\in \frac{\F_q[x]}{(x^n-1)},$
 a linear code $\C$ of length $n$ over $\F_q$ corresponds to a subset of the residue class $\frac{\F_q[x]}{(x^n-1)}.$ The linear code $\C$ is cyclic if and only if the corresponding subset in $\frac{\F_q[x]}{(x^n-1)}$ is an ideal of the ring $\frac{\F_q[x]}{(x^n-1)}.$
 It is well known that every ideal of $\frac{\F_q[x]}{(x^n-1)}$ is principal.
  To distinguish the principal ideal $(g(x))$ of $\F_q[x]$ from that ideal in $\frac{\F_q[x]}{(x^n-1)}$, we use the notation $\langle g(x)\rangle$ for the principal ideal of $\frac{\F_q[x]}{(x^n-1)}$ generator by $g(x)$.
  Let $\C=\langle g(x)\rangle$ be a cyclic code, where $g(x)$ is monic and has the least degree. Let $h(x)=\frac{x^n-1}{g(x)}$. Then $g(x)$ and $h(x)$ are called the generator polynomial and check polynomial, respectively. The dual code, denoted by $\C^{\perp}$, of $C$ has generator polynomial $h^*(x)$, which is the reciprocal of $h(x)$.
Since cyclic codes have efficient encoding and decoding algorithms, they have wide applications in storage and communication systems \cite{Chien,Forney,Prange}. One way of constructing cyclic codes over $\F_q$ with length $n$ is to use the generator polynomial
\begin{equation}\label{gene}
  \frac{x^n-1}{\gcd(S(x),x^n-1)},
\end{equation}
where $S(x)=\sum_{i=0}^{n-1}s_ix^i\in \F_q[x]$ and $s^{\infty}=(s_i)_{i=0}^{\infty}$ is a sequence of period $n$ over $\F_q$.
We call the cyclic code $\C_s$ with the generator polynomial of (\ref{gene}) the code defined by the sequence $s^{\infty}$, and the sequence $s^{\infty}$ the defining sequence of the cyclic code $\C_s$. In the last decade, impressive progress has been made in constructing cyclic codes using this approach \cite{DIT1,DSIAM,DFFA1,DB1,DZ,TQX}.

 In \cite{DSIAM,DZ}, Ding and Zhou constructed some cyclic codes and posed some open problems. So far, some open problems have been solved or partially solved. Notably,  Tang {\it {et al.}} \cite{TQX} solved the Open problem 4 proposed in \cite{DZ} and the Open problem 5.29 proposed in \cite{DSIAM};
  Li {\it {et al.}} \cite{LZLK} partially solved the Open problem 3 proposed in \cite{DZ};
Rajabi and Khashyarmanesh \cite{RK} solved the Open problems 5.26 and 5.30 proposed in \cite{DSIAM}.

This article deals with cyclic codes from functions introduced in symmetric cryptography. An excellent book devoted to Boolean functions for cryptography and coding theory is due to Carlet \cite{Book-Carlet}. A survey on linear codes from cryptographic functions, including open problems, is \cite{Li-Mesnager-2020}.
A recent handbook's chapter covering general linear codes from functions, including several developments and constructions, is \cite{Mesnager-Handbook}.

 More specifically, we follow fascinating articles from Ding and Zhou by providing more achievements in this theme through new results, including refinements and corrections of some known results on cyclic codes from functions. We shall focus more precisely on  cyclic codes designed  by employing specific functions listed below.

 \begin{itemize}
             \item Cyclic codes from the Gold function (we complement some results, see Subsection \ref{Gold});
             \item Binary cyclic codes from the Kasami function (we  provide some correct results on the parameters of the cyclic codes given by \cite{DZ} and give some answers to Open Problem 1 proposed in \cite{DZ}, see Subsection  \ref{Kasami});
             \item Binary cyclic codes from Bracken-Leander function (we give some answers to Open problems 5.16 and 5.25 raised  in \cite{DSIAM} and complements to \cite[Corollary 5.15]{DSIAM}, see Subsection \ref{Bracken-Leander}).

           \end{itemize}

The paper is organized as follows. Section \ref{Preliminaries} fixes our notation and recalls some basic concepts on functions and linear codes accompanied with some related results needed in the paper. In Section \ref{Generic-Construction}, we start by presenting  in Subsection \ref{Known-results}  nice design method  initiated by Ding  and fascinating derived  results concerning  the constructions of cyclic codes  and related sequences from functions. Next, in Subsections \ref{Gold}, \ref{Kasami} and \ref{Bracken-Leander}, we study more in-depth former results dealing with design cyclic codes from Gold, Kasami, and  Bracken-Leander functions, respectively, as in the lines described above. Section \ref{Conclusion}  concludes the paper by summiting our main contributions briefly.

\section{Notation and preliminaries}\label{Preliminaries}
In this paper, we shall employ several (different but somehow related) ingredients from coding theory, sequences, and the theory of vectorial (Boolean) functions (namely, highly nonlinear functions or almost nonlinear functions). We first present some basic notation employed in subsequent sections and state preliminary results on elements from the cyclotomic field theory, coding theory, linear feedback shift registers structures and related sequences, and specific vectorial functions. All these ingredients will be helpful throughout this paper.

\subsection{Some notation fixed throughout this paper}\label{notation}
Throughout this paper, we assume that $\gcd(n,q)=1$ and  adopt the following notation unless otherwise stated.
\begin{itemize}
\item For any finite  $E$, $|E|$  will
denote the cardinality of $E$.
  \item $\mathbb {F}_{q}$ denotes the Galois field of order $q$.
\item Given a finite field $\mathbb{F}_q$,  $\mathbb{F}_q^{\star}$ denotes the multiplicative group $(\mathbb{F}_q \setminus\{0\},\times)$.
  \item Let $m$ and $s$ be two positive integers and let $p$ be a prime. We set $q:=p^s$, $r:=q^m$, and $n:=r-1$.
  \item $\mathbb{Z}_n=\{0,1,2,\ldots,n-1\}$, the ring of integers modulo $n$.
  \item $\mathbb{N}_q(x)$ is a function defined by $\mathbb{N}_q(x)=0$ if $x\equiv0({\rm{mod~}}q)$ and $\mathbb{N}_q(x)=1$, otherwise.
  \item $\alpha$ is a primitive element of $\F_{r}$.
  \item $m_a(x)$ denotes the minimal polynomial of $a\in \F_{r}$ over $\F_q$  (where $r$ is a divisor of $q$, both prime powers of the same prime).
  \item ${\rm{Tr}}$ denotes the (absolute) trace function from $\F_{2^m}$ to $\F_2$.
  \item ${\rm{Tr}}_{r/q}$ denotes the (relative) trace function from $\F_{r}$ to $\F_q$.
  \item $C_i$ denotes the $q$-cyclotomic coset modulo $n$ containing $i$.
  \item $\Gamma$ is the set of all coset leaders of the $q$-cyclotomic cosets modulo $n$.
    \item ${A\choose t}$ is the set of all $t$-subsets of $A$, where $A$ is a set and $t$ is a non-negative integer.
\end{itemize}
 By Database we mean the collection of the tables of best linear codes known maintained by Markus Grassl at http://www.codetables.de/.

\subsection{On the $q$-cyclotomic cosets modulo $n$}
A $q$-cyclotomic coset $C_s$ modulo $n$ is defined to be
$$C_s=\{s,sq,sq^2,\ldots,sq^{l_s}\} {\rm{mod~}}n \subset \mathbb{Z}_n$$
where $l_s$ is the smallest positive integer such that $sq^{l_s}\equiv s({\rm mod~}n)$, and is called the size of $C_s$.
The smallest integer in $C_j$ is called the coset leader of $C_j$. Let $\Gamma$ denote the set of all coset leaders.
 It is well known that
$$\prod_{j\in C_i}(x-\alpha^j)$$
is an irreducible polynomial of degree $l_i$ over $\F_q$ and is the minimal polynomial of $\alpha^i$ over $\F_q$.
 Moreover, we have
$$\bigcup_{j \in \Gamma}C_j=\mathbb{Z}_n$$
and
$$x^n-1=\prod_{i\in \Gamma}\prod_{j\in C_i}(x-\alpha^j).$$
For a non-negative integer $s$, the $q$-adic expansion of $s$ is defined as
$$s=i_0+i_1q+\cdots+i_{m-1}q^{m-1}$$
with $0\leq i_0,i_1,\ldots,i_{m-1}\leq q-1$, and can be simplified to the sequence of the form
$$s=(i_0,i_1,\ldots,i_{m-1}).$$

\subsection{On linear feedback shift register}
Let $s^L=s_0s_1\cdots s_{{L-1}}$ be a sequence over $\F_q$. The linear span of $s^{L}$ is defined to be the smallest positive integer $l$ such that there are constants $c_0=1, c_1,\ldots,c_l\in \F_q$ such that
$$-c_0s_i=c_1s_{i-1}+c_2s_{i-2}+\cdots+c_ls_{i-l}$$
for all $l\leq i\leq L-1.$

A shift register is converted into a code generator by including a feedback loop, which computes a new term for the left-most stage based on the $n$ previous terms. The $n$ $q$-ary storage elements $a_i$ are called the stages of the shift register, and their contents
$\xi_i=(a_i, a_{i+1}, \ldots, a_{i+n-1})$ are called the states of the shift register. The shift register is run by an external clock which generates a timing signal every $t_0$ seconds. A delay element stores one bit (from some alphabet) for one clock cycle, after which the bit is pushed out and replaced by another. A linear shift register is a series of delay elements; a bit enters at one end of the shift register and moves to the next delay element with each new clock cycle. A linear feedback shift register (LFSR for short) is a linear shift register in which the output is fed back into the shift register as part of the input. The polynomial
$$c(x)=\sum_{i=0}^{l}c_ix^i$$
is called the feedback polynomial of the shortest LFSR that generates $s^L$. Such an integer always exists for finite sequences. Moreover, any feedback polynomial of $s^{\infty}$ is called a characteristic polynomial. The characteristic polynomial with the smallest degree is the minimal polynomial of the periodic sequence $s^{\infty}$.

The following lemma gives a way to determine the related linear span and minimal polynomial of a periodic sequence $s^{\infty}$ \cite{AB}.
\begin{lemma}\label{lem1.0.1}
Any sequence $s^{\infty}$ over $\F_q$ of period $q^m-1$ has a unique expansion of the form
$$s_t=\sum_{i=0}^{q^m-2}c_i\alpha^{it}$$
for all $t\geq 0$, where $c_i\in\F_{q^m}.$ Let index set be $I=\{i:c_i\neq 0\}$, then the minimal polynomial $\mathbb{M}_s(x)=\prod_{i\in I}(1-\alpha^ix),$ and the linear span of $s^{\infty}$ is $|I|.$
\end{lemma}
\subsection{Some  background from coding theory}
The \emph{Hamming weight} of a vector ${\bar a}=(a_0,a_1, \cdots,a_{n-1})$ of $\mathbb{F}_q$, denoted by $\wt({\bar a})$, is the cardinality of its \emph{support} defined as $$Supp({\bar a})=\{0\leq i\leq n-1: a_i\neq 0\}.$$ An $[n, k, d]$ linear code $\mathcal C$ over $\mathbb{F}_p$ is a $k$-dimensional vector subspace of the $n$-dimensional vector space $\mathbb{F}_{p^n}$ with minimum (nonzero) Hamming weight $d$.  An element of $\mathcal C$ is said to be a \emph{codeword}.
A linear code $\mathcal C$ of length $n$ over $\mathbb F_{q}$ its (Euclidean) dual code is denoted by $\mathcal C^ {\perp}$ and defined as
$$\mathcal C^ {\perp}=\left\{  (b_0, b_1, \cdots, b_{n-1})\in \mathbb F_{q}^n: \sum _{i=0}^{n-1} b_i c_i=0, \forall (c_0, c_1, \cdots, c_{n-1}) \in \mathcal C\right\}.
$$

\subsection{Differential uniformity and the related perfect nonlinear and planar functions}
Let $q$ be prime power, $m$  be a positive integer,  and $F$ be a function from $\F_q^m$ to itself and $a,b \in \F_{q^m}$. The {\it Difference Distribution Table} (DDT for short) of $F$ is given by $q^m\times q^m$ table ${\rm{DDT}}_f$, in which the entry for the $(a,b)$ position is given by
$${\rm{DDT}}_F(a,b)=\#\{x\in\F_{q^m}|F(x+a)-F(x)=b\}.$$
For any $a,b\in \F_{q^m}, a\neq 0,$ the value
$$\delta_F=\max_{a,b\in\F_{q^m}, a\neq 0}\#\{x\in \F_{q^m}|F(x+a)-F(x)=b\}$$
is called the {\it differential uniformity} of $f(x)$. In this case, we say that $F$ an $\delta_F$-uniform DDT function.
When $F$ is used as an S-box inside a cryptosystem, the smaller the value $\delta_F$ is the better $F$ to the resistance against differential attack. Typically, the optimal functions  $F$ satisfy $\delta_F=2$ are called {\it almost perfect nonlinear} (APN for short). Besides, if $\delta_F=1$ then $F$  is  said to be a {\it perfect nonlinear} or {\it planar}. No perfect nonlinear function exists on $\F_{2^m}$. Both planar and APN functions over $\F_q^m$ for odd $q$ exist.

In this paper, we shall concentrate  on the classical functions which are listed in the following:
\begin{itemize}
  \item[(i)] Inverse function $F_1(x)=x^{-1}$ over $\F_{2^m}$;
  \item[(ii)] Gold function $F_2(x)=x^{q^h+1}$ over $\F_{2^m}$;
  \item[(iii)] Welch function $F_3(x)=x^{2^h+3}$ over $\F_{2^m}$;
  \item[(iv)] Kasami function $F_4(x)=x^{2^{2h}-2^h+1}$ over $\F_{2^{m}}$;
  \item[(v)] Niho-1 function $F_5(x)=x^{2^h+2^{\frac{h}{2}}-1}$ over $\F_{2^m}$;
  \item[(vi)] Niho-2 function $F_6(x)=x^{2^h+2^{\frac{3h+1}{2}}-1}$ over $\F_{2^m}$;
  \item[(vii)] Dobbertin function $F_7(x)=x^{2^{4h}+2^{3h}+2^{2h}+2^h-1}$ over $\F_{2^m}$;
  \item[(viii)] Bracken-Leander function $F_8(x)=x^{2^{2h}+2^{h}+1}$ over $\F_2$.
\end{itemize}
These functions have low differential uniformity under some restrictions:
\begin{itemize}
  \item The inverse function $F_1(x)$ is APN if $m$ is odd; $F_1(x)$ is a $4$-uniform DDT function if $m$ is even (\cite{BD,Gold,Ny}).
  \item In the binary case, the Gold function $F_2(x)=x^{2^h+1}$ is APN if $m$ is odd and $\gcd(h,m)=1$; $F_2(x)$ is a $4$-uniform DDT function if $m\equiv2({\rm mod~}4)$ and $\gcd(h,m)=2$. In the non-binary case, the Gold function $F_2(x)=x^{q^h+1}$ is planar if both  $\frac{m}{\gcd(m,h)}$ and $q$ are odd (\cite{Gold,Ny}).
  \item The Welch function $F_3(x)=x^{2^{h}+3}$ is APN if $m$ is odd and $h=\frac{m-1}{2}$(\cite{Dobb2}). The Niho-1 function $F_5(x)$ is APN if $h$ is even and $m=2h+1$; The Niho-2 function is APN if $h$ is odd and $m=2h+1$\cite{Dobb3}. The Dobbertin function $F_7(x)$ is APN if $m$ is odd and $m=5h$ (\cite{Dobb4}). The above contributions are all due to  Dobbertin.
  \item The Kasami function $F_4(x)$ is APN if $m$ is odd and $\gcd(m,h)=1$ (\cite{Kasami1}).
  \item The Bracken-Leander function $F_8(x)$ is a $4$-uniform DDT function if $m=4h$ and $h$ is odd (\cite{BL,Dobb}).
\end{itemize}
For more literature about  these cryptographic concepts and related constructions  over finite fields of even characteristic, readers can refer to  the book \cite{Book-Carlet}.

\section{On a crucial generic construction of cyclic codes  as sequences from polynomials}\label{Generic-Construction}

Given any polynomial $f(x)$ over $\F_{q^m}$, we define its associated sequence $s^{\infty}$ by
\begin{equation}\label{infty}
 s_i={\rm{Tr}}_{q^m/q}(f(\alpha^i+1))
\end{equation}
for all $i\geq 0$. The code $\C_s$ defined by the sequence $s^{\infty}$ in (\ref{infty}) is called the code from the polynomial $f(x)$ for simplicity.
\subsection{The known results}\label{Known-results}
\begin{itemize}
  \item Ding studied the binary cyclic codes $\C_s$ from inverse function for any $m$ in \cite{DSIAM}. Note that in the non-binary case, the dimension of the code $\C_s$ over $\F_q$ was settled in \cite{TQX}.
  \item Ding and Zhou \cite{DZ} studied the binary cyclic codes from the Welch function $F_3(x)=x^{2h+3}$, where $m=2h+1\geq 7$.
  \item Ding and Zhou \cite{DZ} studied the binary cyclic codes from the Kasami function $F_4(x)$, where $\gcd(m,h)=1$ and
      \begin{equation*}\label{restriction1}
  1\leq h
 \leq \left\{
    \begin{array}{ll}
       \frac{m-1}{4}, & \hbox{if $m\equiv1 ({\rm{mod~}}4)$;} \\
       \frac{m-2}{4}, & \hbox{if $m\equiv2 ({\rm{mod~}}4)$;} \\
       \frac{m-3}{4}, & \hbox{if $m\equiv3 ({\rm{mod~}}4)$;} \\
       \frac{m-4}{4}, & \hbox{if $m\equiv0 ({\rm{mod~}}4)$.}
    \end{array}
  \right.
\end{equation*}
  \item Ding and Zhou \cite{DZ} considered the binary cyclic codes $\C_s$ from the Niho-1 APN function; Li {\it{et al.}} \cite{LZLK} considered the binary cyclic codes $\C_s$ from the Niho-2 APN function.
  \item Tang {\it{et al.}} \cite{TQX} studied the binary cyclic codes from the Dobbertin APN function.
\end{itemize}

\subsection{Binary cyclic codes from the Gold function $F_2(x)$}\label{Gold}
This subsection considers the binary cyclic codes from the Gold function $F_2(x)$.
Let $s^{\infty}$ be the sequence of (\ref{infty}), where $f(x)=F_2(x)$. Then
\begin{eqnarray}\label{eee0}
 \nonumber s_t &=& {\rm{Tr}}\left((\alpha^t+1)^{2^h+1}\right) \\
 \nonumber  ~ &=& {\rm{Tr}}\left((\alpha^t)^{2^h+1}+(\alpha^t)^{2^h}+\alpha^t+1\right) \\
  ~ &=& {\rm{Tr}}\left((\alpha^t)^{2^h+1}\right)+{\rm{Tr}}(1).
\end{eqnarray}
\begin{lemma}\label{lem1.1}
Let $m$ and $h$ be two positive integers and let $q$ be a prime power. Then we have
\begin{equation*}\label{lemeq1.1}
  \gcd\left(q^{h}+1, \frac{q^m-1}{q^{\gcd(m,h)}-1}\right)
  =\frac{\gcd\left(q^{2h}-1,q^m-1\right)}
  {\gcd\left(q^h-1,q^m-1\right)}.
\end{equation*}
\end{lemma}
\begin{proof}
It is easily seen that
$\gcd(q^m-1,q^h-1)=q^{\gcd(m,h)-1}.$
Then we have the following
\begin{eqnarray*}
  \gcd(q^{2h}-1,q^m-1)
   &=& (q^{\gcd(m,h)}-1)\gcd\left((q^h+1)\cdot\frac{q^h-1}{q^{\gcd(m,h)}-1},
  \frac{q^m-1}{q^{\gcd(m,h)}-1}\right) \\
  ~ &=& \gcd(q^h-1,q^m-1)\gcd\left((q^h+1),
  \frac{q^m-1}{q^{\gcd(m,h)}-1}\right).
\end{eqnarray*}
 This completes the proof.
\end{proof}
If $F_2(x)$ is APN, we have
\begin{lemma}\label{lem01}
Let $s^{\infty}$ be the sequence of (\ref{infty}), where $f(x)=F_2(x)$, $m$ is odd and $\gcd(m,h)=1.$ Then the linear span $\mathbb{L}_s$ of $s^{\infty}$ is equal to $m+1$ and the minimal polynomial $\mathbb{M}_s(x)$ of $s^{\infty}$ is given by
\begin{equation*}\label{Ms2.0}
  \mathbb{M}_s(x)=(x-1)m_{\alpha^{-(2^h+1)}}(x).
\end{equation*}
\end{lemma}
\begin{proof}
Since $m$ is odd, ${\rm{Tr}}(1)=1$.
According to Lemma \ref{lem1.1}, the size of $C_{1+2^h}$ equals $m$ since
\begin{eqnarray*}
  \gcd(2^m-1,2^h+1) = \gcd\left(2^h+1,\frac{2^m-1}
{2^{\gcd(m,h)}-1}\right)=
\frac{2^{\gcd(2h,m)}-1}{2^{\gcd(h,m)}-1}
=1.
\end{eqnarray*} Therefore, the desired conclusions follows.
\end{proof}
If $F_2(x)$ is a $4$-uniform DDT function, then we have
\begin{lemma}\label{lem1}
Let $s^{\infty}$ be the sequence of (\ref{infty}), where $f(x)=F_2(x)$, $m\equiv2({\rm mod~}4)$ and $\gcd(m,h)=2.$ Then the linear span $\mathbb{L}_s$ of $s^{\infty}$ is equal to $m$ and the minimal polynomial $\mathbb{M}_s$ of $s^{\infty}$ is given by
\begin{equation*}\label{Ms2}
  \mathbb{M}_s(x)=m_{\alpha^{-(2^h+1)}}(x).
\end{equation*}
\end{lemma}
\begin{proof}
Since $m$ is even and $\gcd(m,h)=2$, we have ${\rm{Tr}}(1)=0$, $2^{\gcd(m,h)}-1=3$.
 Combining Lemma \ref{lem1.1} and the fact that $\gcd(2^h+1,2^{\gcd(m,h)}-1)=\gcd(2^h+1,3)=1$, we have
\begin{eqnarray*}
  \gcd(2^h+1,2^m-1) &=& \gcd\left(2^h+1,
\frac{2^m-1}{3}\right) \\
  ~ &=& \gcd\left(2^h+1,
\frac{2^m-1}{2^{\gcd(m,h)}-1}\right) \\
  ~ &=& \frac{2^{\gcd(2h,m)}-1}{2^{\gcd(h,m)}-1}\\
  ~ &=& 1.
\end{eqnarray*}

Therefore, the size of the $q$-cyclotomic coset containing $2^h+1$ is $m$. The desired results follow from Lemma \ref{lem1.0.1} and Equation (\ref{eee0}).
\end{proof}
\begin{theorem}
The code $\C_s$ defined by the sequence of Lemma \ref{lem01} has parameters $[2^m-1,2^m-2-m,4]$, which is equivalent to the binary extended Hamming code $\widehat{\mathcal{H}_m}$.
\end{theorem}
\begin{proof}
The dimension of $\C_s$ follows from Lemma \ref{lem01} and the definition of the code $\C_s$. According to the Sphere Packing bound (see, e.g.
\cite{HufPless-2003}) and BCH bound (see, e.g. \cite{HufPless-2003,KT69}), we know that $2\leq d(\C_s)\leq 4.$ Since the parity-check matrix of $\C_s$ is
$$\left(
    \begin{array}{cccc}
      1 & 1 & \cdots & 1 \\
      \alpha^{2^h+1} & (\alpha^{2^h+1})^2 & \cdots & (\alpha^{2^h+1})^{2^m-1}
    \end{array}
  \right),
$$
the minimum distance $d(\C_s)$ has to be even and $d(\C_s)\neq 2$. Therefore, $d(\C_s)=4.$
\end{proof}
\begin{theorem}
The code $\C_s$ defined by the sequence of Lemma \ref{lem1} has parameters $[2^m-1,2^m-1-m,3]$, which is equivalent to the binary Hamming code $\mathcal{H}_m$.
\end{theorem}
\begin{proof}
The dimension of $\C_s$ follows from Lemma \ref{lem1} and the definition of the code $\C_s$. Combining the BCH and the Sphere Packing bounds, we have $2\leq d(\C_s)\leq 3$. It is easy to check that there is no codeword $\mathbf{c}$ in $\C_s$ satisfies $w_H(\mathbf{c})=2$, then the minimum distance of $\C_s$ equals $3$. It is well-known that the binary code with parameters $[2^m-1,2^m-m-1,3]$ is equivalent to the binary Hamming code $\mathcal{H}_m$.
\end{proof}
\begin{remark}
Let $s^{\infty}$ be the sequence of (\ref{infty}), where $f(x)=F_2(x)=x^{q^h+1}$ over $\F_q$, where both $m$ and $q$ are odd. Ding \cite{DSIAM} considered the cyclic code $\C_s$ over $\F_q$ form monomial $F_2(x)=x^{q^h+1}$ and gave the linear span $\mathbb{L}_s$ and minimal polynomial $\mathbb{M}_s(x)$ of $s^{\infty}$ in \cite[Lemma 5.1]{DSIAM}. In the proof of \cite[Lemma 5.1]{DSIAM}, the formula
$$\gcd(q^{\kappa}+1,q^m-1)=\frac{\gcd(q^{2\kappa}-1,q^m-1)}
{\gcd(q^{\kappa}-1,q^m-1)}=
\frac{q^{\gcd(2\kappa,m)}-1}{q^{\gcd(\kappa,m)}-1}=2
$$
should be replaced by
$$\gcd(q^{\kappa}+1,q^m-1)=2,$$
which follows directly from Lemma \ref{lem31} as $q$ and $m/\gcd(m,\kappa)$ are odd. From \cite[Lemma 2.1]{DingTor}, the size of $C_{q^{\kappa}+1}$ is $m$. In fact, the result in \cite[Lemma 5.1]{DSIAM} holds when $q$ and $m/\gcd(m,\kappa)$ are odd. In other words, $m$ can be even.
\end{remark}

\subsection{Binary cyclic codes from the Kasami function $F_4(x)$}\label{Kasami}
As is well-known, $\delta_{F_4}=2$ if $\gcd(m,h)=1$; $\delta_{F_4}=4$ if $m\equiv2({\rm mod~}4)$ and $\gcd(m,h)=2$. In this subsection, we consider the cyclic codes defined by the monomial $F_4(x)$ over $\F_2$ and assume that \begin{equation*}\label{consub6}
  \left\{
    \begin{array}{ll}
      1\leq h\leq\frac{m-1}{2}, & \hbox{if $m\equiv1 ({\rm{mod~}}4)$;} \\
      1\leq h\leq\frac{m-2}{2}, & \hbox{if $m\equiv2 ({\rm{mod~}}4)$;} \\
      1\leq h\leq\frac{m-3}{2}, & \hbox{if $m\equiv3 ({\rm{mod~}}4)$;} \\
      1\leq h\leq\frac{m-4}{2}, & \hbox{if $m\equiv0 ({\rm{mod~}}4)$.}
    \end{array}
  \right.
\end{equation*}

Note that this paper has no restriction on the value of $\gcd(m,h)$.

We define the following notation, which will be used in subsequent results and their respective proofs. For any integer $j$ with $0\leq i\leq 2^m-1$, the $2$-weight of $i$, denoted by $wt(i)$, is defined as the number of nonzero coefficients in its $2$-adic expansion:
$$i=a_0+a_1\cdot2+\cdots+a_{m-1}\cdot 2^{m-1},$$
where $a_i\in \F_2.$
Let $A$ and $B$ be the two sets defined as follows, respectively.
\begin{eqnarray}\label{A1A}
  \nonumber A &=&  \{0,1,2,\ldots,2^h-1\}\\
  ~ &=& \left\{i=(a_0,a_1,\ldots,a_{h-1},\underbrace{0,\ldots,0}_{m-h}):a_i\in\F_2 {\rm{~and~~}}0\leq i\leq h-1\right\}
\end{eqnarray}
and
\begin{eqnarray}\label{B1B}
  \nonumber B &=&  2^{m-h}+A\\
  ~ &=& \left\{i=(a_0,a_1,\ldots,a_{h-1},\underbrace{0,\ldots,0}_{m-2h},1,\underbrace{0,\ldots,0}_{h-1}):a_i\in\F_2; 0\leq i\leq h-1\right\}.
\end{eqnarray}

For convenience, we use the vector form more frequently to represent the elements of $A$ and $B$ in subsequent proofs.
 Moreover,  given a positive integer $t$, we define the associate sets $A_t$ and $B_t$ as follows.
$$A_t=\{i: i\in A{\rm{~and~}}wt(i)=t\}$$
and
$$B_t=\{j: j\in B{\rm{~and~}}wt(j)=t\}.$$
Besides, let $u_{(i)}$ be the least non-negative integer such that $i\leq 2^{u_{(i)}}-1$ and $v_{(i)}$ be the largest non-negative integer such that $2^{v_{(i)}}|i$. Clearly, one has  $u_{(i)}> v_{(i)}.$ Next, we shall follow the notation given in \cite{DZ}, which will be used to determine the linear span of $s^{\infty}$.
Let $t$ be a positive integer.
We define $T=2^t-1$.
For any odd $a\in\{1,2,3,\ldots,T\}$, define
\begin{equation*}
  \epsilon_a^{(t)}=\left\{
                     \begin{array}{ll}
                       1, & \hbox{if $a=2^h-1$,} \\
                       \left\lceil\log_2\frac{T}{a}\right\rceil, & \hbox{if $1<a<2^h-1$,}
                     \end{array}
                   \right.
\end{equation*}
and $\kappa_a^{(t)}=\epsilon_a^{(t)} ({\rm mod~}2).$
The following lemma will be useful in the sequel.
\begin{lemma}\label{DZlem9}\cite{DZ}
Let $N_t$ denote the total number of odd $\epsilon_a^{(t)}$ when $a$ ranges over all odd numbers in the set $\{1,2,\ldots,T\}.$ Then $N_1=1$ and $N_t=\frac{2^t+(-1)^{t-1}}{3}$
for all $t\geq 2$.
\end{lemma}

Observe that
\begin{eqnarray*}
  \nonumber{\rm{Tr}}(F_4(x+1)) &=& {\rm{Tr}}\left((x+1)(x+1)^{\sum_{i=0}^{h-1}2^h+i}\right) \\
  \nonumber~ &=& {\rm{Tr}}\left((x+1)\prod_{i=0}^{h-1}(x^{2^{h+i}}+1)\right) \\
  \nonumber~ &=& {\rm{Tr}}\left(\sum_{i=0}^{2^h-1}x^{i\cdot2^h+1}
  +\sum_{j=1}^{2^h-1}x^j\right)+ {\rm{Tr}}(1).
\end{eqnarray*}
The sequence $s^{\infty}$ of (\ref{infty}) defined by the monomial $F_4(x)$ is then given by
\begin{equation}\label{st1}
  s_t={\rm{Tr}}\left(\sum_{i=0}^{2^h-1}
  (\alpha^t)^{i+2^{m-h}}+\sum_{j=1}^{2^h-1}(\alpha^t)^j\right)+ {\rm{Tr}}(1)
\end{equation} for all $t\geq 0$.

Ding and Zhou (\cite{DZ}) have studied the binary cyclic codes from $F_4(x)$ when $\gcd(m,h)=1$ and $h$ satisfies one of the following conditions:
\begin{equation}\label{restriction1}
  1\leq h
 \leq \left\{
    \begin{array}{ll}
       \frac{m-1}{4}, & \hbox{if $m\equiv1 ({\rm{mod~}}4)$;} \\
       \frac{m-2}{4}, & \hbox{if $m\equiv2 ({\rm{mod~}}4)$;} \\
       \frac{m-3}{4}, & \hbox{if $m\equiv3 ({\rm{mod~}}4)$;} \\
       \frac{m-4}{4}, & \hbox{if $m\equiv0 ({\rm{mod~}}4)$.}
    \end{array}
  \right.
\end{equation}
\begin{lemma}\cite[Lemma 22]{DZ}\label{lem4}
Let $h$ be a positive  integer  satisfying Condition {\rm {(\ref{restriction1})}}. Also, let $s^{\infty}$ be the sequence of {\rm (\ref{st1})}. Then the linear span $\mathbb{L}_s$ of $s^{\infty}$ is given by
\begin{equation*}\label{Ls1}
  \mathbb{L}_s=\left\{
                 \begin{array}{ll}
                   \frac{m(2^{h+2}+(-1)^{h-1})+{\color{blue}{3}}}{3}, & \hbox{if $h$ is even;} \\
                   \frac{m(2^{h+2}+(-1)^{h-1}-6)+{\color{blue}{3}}}{3}, & \hbox{if $h$ is odd.}
                 \end{array}
               \right.
\end{equation*}
Moreover, we have
\begin{equation*}\label{Ms1}
  \mathbb{M}_s(x)=\left\{
                    \begin{array}{ll}
                      {\color{blue}{(x-1)}}\prod\limits_{i=0}^{2^h-1}m_{\alpha^{-i-2^{m-h}}}(x)
                       \prod\limits_{{\tiny\begin{array}{c}
                      {1}\leq 2j+1\leq 2^h-1\\
                      \kappa_{2j+1}^{(h)}=1\end{array}
                      }}m_{\alpha^{-2j-1}}(x),& \hbox{if $h$ is even;} \\
                      {\color{blue}{(x-1)}}\prod\limits_{i=1}^{2^h-1}m_{\alpha^{-i-2^{m-h}}}(x)
                       \prod\limits_{{\tiny\begin{array}{c}
                      {3}\leq 2j+1\leq 2^h-1\\
                      \kappa_{2j+1}^{(h)}=1\end{array}
                      }}m_{\alpha^{-2j-1}}(x),& \hbox{if $h$ is odd.}
                    \end{array}
                  \right.
\end{equation*}
\end{lemma}
Unfortunately, the results in Lemma \ref{lem4} are not rigorous. When $m$ is even, its minimal polynomial does not contain the factor $(x-1)$. The following are our corrected results. Correspondingly, we also correct the lower bound in \cite{DZ} about the minimum distance of its corresponding cyclic code. Since the proofs of the corrected results are exactly the same as that of \cite{DZ}, we omit its here.
\begin{lemma}\label{lem18}
Let $s^{\infty}$ be the sequence of {\rm (\ref{st1})} and let $h$ satisfy Condition {\rm {(\ref{restriction1})}}.
 Then the linear span $\mathbb{L}_s$ of $s^{\infty}$ is given by
\begin{equation*}\label{Ls1}
  \mathbb{L}_s=\left\{
                 \begin{array}{ll}
                   \frac{m(2^{h+2}-1)+
                   {\color{blue}{3\cdot\mathbb{N}_2(m)}}}{3}, & \hbox{if $h$ is even;} \\
                   \frac{m(2^{h+2}-5)+
                   {\color{blue}{3\cdot\mathbb{N}_2(m)}}}{3}, & \hbox{if $h$ is odd.}
                 \end{array}
               \right.
\end{equation*}
Moreover, we have

\begin{small}\begin{equation*}\label{Ms1.1.1}
  \mathbb{M}_s(x)=\left\{
                    \begin{array}{ll}
                      {\color{blue}{(x-1)^{\mathbb{N}_2(m)}}}\prod\limits_{i=0}^{2^h-1}m_{\alpha^{-i-2^{m-h}}}(x)
                       \prod\limits_{{\tiny\begin{array}{c}
                      {3}\leq 2j+1\leq 2^h-1\\
                      \kappa_{2j+1}^{(h)}=1\end{array}
                      }}m_{\alpha^{-2j-1}}(x),& \hbox{if $h$ is even;} \\
                      {\color{blue}{(x-1)^{\mathbb{N}_2(m)}}}\prod\limits_
                      {i=1}^{2^h-1}m_{\alpha^{-i-2^{m-h}}}(x)
                       \prod\limits_{{\tiny\begin{array}{c}
                      {3}\leq 2j+1\leq 2^h-1\\
                      \kappa_{2j+1}^{(h)}=1\end{array}
                      }}m_{\alpha^{-2j-1}}(x),& \hbox{if $h$ is odd.}
                    \end{array}
                  \right.
\end{equation*}
\end{small}
\end{lemma}
\begin{theorem}\label{thm19}
	Let $h$ be a positive  integer  satisfying Condition {\rm {(\ref{restriction1})}}.
	The binary code $\C_s$ defined by the sequence of {\rm{(\ref{st1})}} has parameters $[2^m-1,2^m-1-\mathbb{L}_s,d]$ and generator polynomial $\mathbb{M}_s(x)$, where $\mathbb{L}_s$ and $\mathbb{M}_s(x)$ are given in Lemma \ref{lem18}. In addition,  the minimum  Hamming distance $d$ of $\C_s$ is bounded as follows
	according to the parity of  $h$ and $m$, respectively.
	\begin{equation*}
		d\geq\left\{
		\begin{array}{ll}
			2^h+2, & \hbox{if $h$ is even and $m$ is odd;} \\
			2^h+1, & \hbox{if $h$ is even and $m$ is even;}\\
			2^h,   & \hbox{if $h$ is odd.}
		\end{array}
		\right.
	\end{equation*}
\end{theorem}
\begin{example}
	Let $(m,h)=(9,2)$ and $\alpha$ be a generator of $\F_{2^m}$ with $\alpha^9+\alpha^4+1=0$. Then the generator polynomial of the cyclic code $\C_s$ is $\mathbb{M}_s(x)=x^{46} + x^{45} + x^{42} + x^{41} + x^{40} + x^{39} + x^{38} + x^{37} + x^{36} + x^{35} + x^{31} +
	x^{28} + x^{27} + x^{24} + x^{23} + x^{22} + x^{21} + x^{20} + x^{18} + x^{15} + x^{14} + x^{12} +
	x^9 + x^6 + x^3 + x^2 + x + 1$ and $\C_s$
	is a $[511,465,8]$ binary cyclic code.
\end{example}

An interesting open problem proposed by Ding and Zhou in \cite{DZ} was stated as follows by  considering other restrictions on $h$.

\begin{open}\cite{DZ}\label{DZop1}
Determine the dimension and the minimum distance of the code $\C_s$ when $h$ satisfies
\begin{equation}\label{op1}
  \left\{
    \begin{array}{ll}
      \frac{m+3}{4}\leq h\leq\frac{m-1}{2}, & \hbox{if $m\equiv1 ({\rm{mod~}}4)$;} \\
      \frac{m+2}{4}\leq h\leq\frac{m-2}{2}, & \hbox{if $m\equiv2 ({\rm{mod~}}4)$;} \\
      \frac{m+1}{4}\leq h\leq\frac{m-3}{2}, & \hbox{if $m\equiv3 ({\rm{mod~}}4)$;} \\
      \frac{m}{4}\leq h\leq\frac{m-4}{2}, & \hbox{if $m\equiv0 ({\rm{mod~}}4)$.}
    \end{array}
  \right.
\end{equation}
\end{open}
In the following, we shall  answer to Open Problem \ref{DZop1}. To this end, we need the following lemma, which comes from \cite{DZ}. Although the proof of the following lemma has already been given in \cite{DZ}, we still provide another proof to facilitate understanding subsequent proofs.
\begin{lemma}\label{192021}\cite{DZ}
Let $h$ be a positive integer  satisfying Condition {\rm{(\ref{restriction1})}}. Then  the following facts concerning  the   $2$-cyclotomic coset  $C_i$ modulo $n$ containing $i$ follow.
\begin{itemize}
  \item For any $j \in B$, $|C_j|=m;$
  \item For any pair of distinct $i$ and $j$ in $B$, $C_i\cap C_j=\emptyset.$
\end{itemize}
\end{lemma}
\begin{proof}
According to the definition of $B$, for any integer $j$ in $B$, there exist $a_0,a_1,\ldots,a_{h-1}\in\F_2$ such that
\begin{equation*}\label{form1}
  j=(a_0,a_1,\ldots,a_{h-1},\underbrace{0,
  \ldots,0}_{m-2h},\underline{1},\underbrace{0,\ldots,0}
  _{h-1}).
\end{equation*}
Since $h$ satisfies the conditions of (\ref{restriction1}), we have
\begin{equation}\label{Inq1}
  m-2h>(h-1)+h.
\end{equation}
Recall that $\tau(x_0,x_1,\ldots,x_{n-1})$ denotes the vector $(x_{n-1},x_0, \ldots,x_{n-2})$ obtained from $(x_0,x_1,\ldots,x_{n-1})$ by the cyclic shift of the coordinates $i\mapsto i+1 ({\rm mod~} n)$. Then we obtain
$$2^s\cdot j=\tau^s(a_0,a_1,\ldots,a_{h-1},
\underbrace{0,\ldots,0}_{m-2h},
\underline{1},\underbrace{0,\ldots,0}
  _{h-1}),$$
where $s$ is a non-negative integer. Assume that there exists some $j\in B$ such that $|C_{j}|=m_0<m$, i.e.,
\begin{eqnarray*}
  2^{m_0}\cdot j &=& \tau^{m_0}(a_0,a_1,\ldots,a_{h-1},
\underbrace{0,\ldots,0}_{m-2h},
\underline{1},\underbrace{0,\ldots,0}
  _{h-1}) \\
  ~ &=& (a_0,a_1,\ldots,a_{h-1},
\underbrace{0,\ldots,0}_{m-2h},
\underline{1},\underbrace{0,\ldots,0}
  _{h-1}).
\end{eqnarray*}
  This contradicts Inequality (\ref{Inq1}). Therefore, the first statement holds.\\

Assume now that there exist distinct $i$ and $j$ in $B$ such that $C_i\cap C_j\neq \emptyset$. Then there exists an integer $m_1$ such that
\begin{eqnarray}\label{eq2}
  \nonumber i &=& 2^{m_1}\cdot j=\tau^{m_1}(a_0,a_1,\ldots,a_{h-1},
\underbrace{0,\ldots,0}_{m-2h},
\underline{1},\underbrace{0,\ldots,0}
  _{h-1}) \\
  ~ &=& (b_0,b_1,\ldots,b_{h-1},
\underbrace{0,\ldots,0}_{m-2h},
\underline{1},\underbrace{0,\ldots,0}
  _{h-1}),
\end{eqnarray}
where $b_i \in\F_2.$ The element $i$ has the form (\ref{eq2}) only if $m-2h\leq (h-1)+h$, which contradicts  Inequality (\ref{Inq1}). Therefore, the second statement holds.
\end{proof}
We shall define a set $B^{*}$ constructed from $u_{(i)}$ and $v_{(i)}$ as follows.
\begin{equation*}\label{Bstar}
  B^*=\{i+2^{m-h}: i \in A, v_{(i)}\geq m-2h+1 {\rm{~~and~~}} u_{(i)}\leq 3h-m-2\}.
\end{equation*}
We claim the following.
\begin{lemma}\label{lem13}
If $h$ satisfies the  following conditions
$$m+2\leq 3h\leq
\left\{
    \begin{array}{ll}
      \frac{3(m-1)}{2}, & \hbox{if $m\equiv1 ({\rm{mod~}}4)$;} \\
      \frac{3(m-2)}{2}, & \hbox{if $m\equiv2 ({\rm{mod~}}4)$;} \\
      \frac{3(m-3)}{2}, & \hbox{if $m\equiv3 ({\rm{mod~}}4)$;} \\
      \frac{3(m-4)}{2}, & \hbox{if $m\equiv0 ({\rm{mod~}}4)$,}
    \end{array}
  \right.
$$
then for any $i\in A$ and odd $j\in A$, we have $C_{i+2^{m-h}}\cap C_j\neq \emptyset$ only if $i+2^{m-h}\in B^*$.
\end{lemma}
\begin{proof}

According to the assumption, one gets $m-2h+1\leq h-1$ and $B^*\subseteq B$. Combining the sequences (\ref{A1A}) and (\ref{B1B}), we deduce that there exists a positive integer  $i\in A$ and an  odd integer $j\in A$  provided that  $v_{(i)}\geq m-2h+1 {\rm{~~and~~}} u_{(i)}\leq 3h-m-2$. This completes the proof.
\end{proof}
\begin{remark}
Note that the assumption  made that $m-2h+1\leq h-1$ is necessary to have $C_{i+2^{m-h}}\cap C_j\neq \emptyset$. However, if $m-2h+1> h-1$, then both conditions that $i \in A$  and $v_{(i)}\geq m-2h+1$ cannot be satisfied simultaneously.
\end{remark}
The following lemma is useful, which extends the condition in \cite[Lemma 21]{DZ} made on $h$, that is, $h$ satisfies Condition  {\rm {(\ref{restriction1})}} by assuming that $h$ satisfies  $5\leq 5h< 2m+3$ (which is a less restrictive assumption).  Therefore, we have the following result.
\begin{lemma}\label{lem15}
Let $h$ be a positive integer  satisfying $5\leq 5h< 2m+3$. Then for any $i+2^{m-h}\in B$ and an odd integer $j\in A$ we have
\begin{equation*}\label{lem21}
  C_{i+2^{m-h}}\cap C_j=\left\{
                          \begin{array}{ll}
                            C_j, & \hbox{if $(i,j)=(0,1)$;} \\
                            \emptyset, & \hbox{otherwise.}
                          \end{array}
                        \right.
\end{equation*}
\end{lemma}
\begin{proof}
It is not difficult to check that $C_{i+2^{m-h}}\cap C_{j}=C_j$ if $(i,j)=(0,1)$.
Let us first prove the desired results hold if $3\leq 3h\leq m+1$.
 Assume that there are $i$ and odd $j$ such that $C_{i+2^{m-h}}\cap C_{j}\neq \emptyset$, where $i+2^{2^{m-h}}\in B$, $j\in A$, and $(i,j)\neq (0,1)$.
According to the definition of $A$ and $B$, $i$ and $j$ are of the form
\begin{equation*}
  i+2^{m-h}=(a_0,a_1,\ldots,a_{h-1},\underbrace{0,\ldots,0}_{m-2h},1,\underbrace{0,\ldots,0}
  _{h-1})
\end{equation*}
and
\begin{equation*}
  j=(a_0,a_1,\ldots,a_{h-1},
  \underbrace{0,\ldots,0}_{m-h}),
\end{equation*}
respectively. Then, $C_{i+2^{m-h}}\cap C_{j}\neq \emptyset$ only if there exist at least $m-2h$ successive zeros between $a_0$ and $a_{h-1}$, i.e., $h-2\geq m-2h$. It contradicts  the assumption made. Therefore, the conclusion follows when $3\leq 3h\leq m+1$.

If $3h\geq m+2$ and $5h< 2m+3$, we have $m-2h+1\leq h-1$ and $m-2h+1>3h-m-2$. From the definition of $B^*$,  the  former set is empty if $m-2h+1>3h-m-2$, i.e., $v_{(i)}>u_{(i)}$. The desired conclusion follows then from Lemma \ref{lem13}.
\end{proof}

We claim that when $h$ satisfies Condition {\rm{(\ref{op1})}}, the size of $C_{i+2^{m-h}}$ can not equal $m$.  The following lemma is needed.
\begin{lemma}\label{lem16}
Let $h$  be a positive integer satisfying Condition  {\rm{(\ref{op1})}}.
Then $|C_{i+2^{m-h}}|\in \left\{\frac{m}{3},\frac{m}{2},m\right\}.$ Moreover, we have the following results on the size of $C_{i+2^{m-h}}$ for any element $i\in A$.
\begin{itemize}
\item If $m$ satisfies one of the following conditions
\begin{itemize}
  \item[{\rm (i)}] $\gcd(6,m)=1$;
  \item[{\rm (ii)}] $\gcd(6,m)=2$ and $2h\leq\frac{m}{2}$;
  \item[{\rm (iii)}] $\gcd(6,m)=3$ and $h\leq\frac{m}{3}$;
  \item[{\rm (iv)}] $\gcd(6,m)=6$ and $2h\leq\frac{m}{2}$,
\end{itemize}
 then $|C_{i+2^{m-h}}|=m$ for any $i\in A$.

  \item If $\gcd(6,m)=2$ and $2h\geq \frac{m}{2}+1$, then $|C_j|=\frac{m}{2}$ if and only if $j=2^{\frac{m}{2}-h}+2^{m-h}$, and
 $|C_j|=m$, otherwise.
  \item If $\gcd(6,m)=3$ and $2h\geq\frac{2m}{3}+1$, then $|C_j|=\frac{m}{3}$ if and only if $j=2^{\frac{m}{3}-h}+
 2^{\frac{2m}{3}-h}+2^{m-h}$, and $|C_j|=m$, otherwise.
  \item If $\gcd(6,m)=6$ and $2h\geq\frac{2m}{3}+1$, then $|C_j|=\frac{m}{2}$ if and only if $j=2^{\frac{m}{2}-h}+2^{m-h}$, and $|C_j|=\frac{m}{3}$ if and only if $j=2^{\frac{m}{3}-h}+ 2^{\frac{2m}{3}-h}+2^{m-h}$, and $|C_j|=m$, otherwise.
 \item If $\gcd(6,m)=6$ and $\frac{m}{2}+1\leq 2h\leq\frac{2m}{3}$, then $|C_j|=\frac{m}{2}$ if and only if $j=2^{\frac{m}{2}-h}+2^{m-h}$, and $|C_j|=m$, otherwise.
\end{itemize}
\end{lemma}

\begin{proof}
 Recall that $u_{(i)}$ is the least non-negative integer such that $i\leq 2^{u_{(i)}}-1$, where $i\in A$. Let $V_1$ and $V_2$ be two sets defined as
\begin{equation*}\label{V1}
  V_1:=\{i+2^{m-h},2\cdot(i+2^{m-h}),\ldots, 2^{h-1}\cdot(i+2^{m-h})\}
\end{equation*}
and
\begin{equation*}\label{V2}
  V_2:=\{i\cdot2^h+1,2\cdot (i\cdot2^h+1),\ldots,2^{m-h-u_{(i)}}
  \cdot (i\cdot2^h+1)\},
\end{equation*}
respectively. It is easy to check that
\begin{equation*}
  \left\{
    \begin{array}{ll}
      |V_1|=h, \\
      |V_2|=m-h-u_{(i)}, \\
      V_1\subseteq C_{i+2^{m-h}}, \\
      V_1\subseteq C_{i+2^{m-h}}.
    \end{array}
  \right.
\end{equation*}
Then we have
$$|C_{i+2^{m-h}}|\geq \max\{h,m-h-u_{(i)}\}\geq \frac{m}{3}.$$
According to the definition of $2$-cyclotomic coset modulo $2^m-1$, we have $$|C_{i+2^{m-h}}|\in \left\{\frac{m}{3},\frac{m}{2},m\right\}.$$
From the assumption that $\gcd(6,m)=1$, we straightforwardly obtain $$|C_{i+2^{m-h}}|=m.$$
  When $\gcd(6,m)=2$ and $\left|C_{i+2^{m-h}}\right|=\frac{m}{2}$, the sequence form of $i+2^{m-h}$ can be divided into two equal parts, each part being the same. It implies that
$$i+2^{m-h}=(\underbrace{0,\ldots,0}
_{\frac{m}{2}-h},\underline{1}_*,\underbrace{0,\ldots,0}
_{\frac{m}{2}-1},1,\underbrace{0,
\ldots,0}_{h-1})$$
and $i=2^{\frac{m}{2}-h}.$
Note that $\underline{1}_*$ has to be the element in the first $h$ positions if $i+2^{m-h}\in B$, i.e., $\frac{m}{2}-h\leq h-1$ if $i+2^{m-h}\in B$. The rest of the conclusions can be proved similarly.
\end{proof}
Determining the linear span of $s^{\infty}$ also requires figuring out how many elements in set $B$ are in the same $q$-cyclotomic coset modulo $n$. We shall define three sets $C_j^{\prime}$,  $B_2^{\prime}$,  and $B_2^{\prime\prime}$ as follows.
 $$C_j^{\prime}:=\{j^{\prime}:j^{\prime}\in C_j {\rm{~~and~~}}
j^{\prime}\in B\},$$
$$
B_2^{\prime}:=
\left\{
  \begin{array}{ll}
    \left\{2^{m-3h+i}+2^{m-h}: i \in \{1,2,\ldots,4h-m-1\}\right\}, & \hbox{if $m$ is odd,} \\
    \left\{2^{m-3h+i}+2^{m-h}: i \in \{1,2,\ldots,4h-m-1\}\right\}\setminus
    \{2^{2h-
      \frac{m}{2}}+2^{m-h}\}, & \hbox{if $m$ is even,}
  \end{array}
\right.
$$
and
$$
B_2^{\prime\prime}:=
\left\{
  \begin{array}{ll}
    \left\{2^{i}+2^{m-h}: i \in \{0,1,\ldots,m-2h\}\right\}, & \hbox{if $m$ is odd,} \\
    \left\{2^{i}+2^{m-h}: i \in \{0,1,\ldots,m-2h\}\right\}\setminus
    \{2^{
      \frac{m}{2}-h}+2^{m-h}\}, & \hbox{if $m$ is even,}
  \end{array}
\right.
$$
respectively. The  sets defined above will play a crucial role. Indeed, it is easy to verify that
$|B_2^{\prime}|=4h-m-1-\mathbb{N}_2(m+1)$
and $|B_2^{\prime\prime}|
=m-2h-\mathbb{N}_2(m+1).$
Next, we have
\begin{lemma}\label{lem17}
Let $h$ satisfy the conditions of {\rm{(\ref{op1})}}. Then
\begin{itemize}
  \item[{\rm{(i)}}] For any $\{j_1,j_2\}\in{B\backslash B_2
\choose 2}$, $C_{j_1}\cap C_{j_2}=\emptyset$;
  \item[{\rm{(ii)}}] If $|C_j^{\prime}|>1$,  then $wt(j)=2$;
  \item[{\rm{(iii)}}]  If $|C_j^{\prime}|>1$, then $|C_j^{\prime}|=2$;
  \item[{\rm{(iv)}}] $|C_j^{\prime}|=2$ if and only if one of the following holds:
  \begin{itemize}
    \item[{\rm{(1)}}]when $m\geq3h$, $j \in B_2^{\prime}$;
    \item[{\rm{(2)}}] when $m<3h$, $j \in B_2^{\prime\prime}.$
  \end{itemize}
\end{itemize}
\end{lemma}
\begin{proof}
If $wt(i)\neq wt(j)$, then $C_i\neq C_j$. Let $C_j^{\prime}=\{j,j_1,\ldots,j_{t}\}$, where $t\geq 1$. Then we obtain $wt(j)=wt(j_1)=\cdots=wt(j_t).$ Assume that $wt(j)=3$ and
$$j=(\underbrace{0,\ldots,0}_{i_1-1},
1,\underbrace{0,\ldots,0}_{i_2-1-i_1},1,
\underbrace{0,\ldots,0}_{h-i_2},
\underbrace{0,\ldots,0}
_{m-2h},\underline{1},\underbrace{0,\ldots,0}
  _{h-1}).$$
If $|C_j^{\prime}|>1$, then there exists some $j_1\in B$ such that
\begin{eqnarray*}
  j_1 &=& 2^{m-2h+h-i_2+1}\cdot j \\
  ~   &=& \tau^{m-h+1-i_2}(\underbrace{0,\ldots,0}_{i_1-1},
1,\underbrace{0,\ldots,0}_{i_2-1-i_1},1
\underbrace{0,\ldots,0}_{h-i_2},
\underbrace{0,\ldots,0}
_{m-2h},\underline{1},\underbrace{0,\ldots,0}
  _{h-1}).
\end{eqnarray*}
The element $j_1\in B$ only if
\begin{equation}\label{cond1}
  \left\{
     \begin{array}{ll}
       h-i_2+m-2h+1\geq h, \\
       i_1+m-2h+1+h-i_2\leq h.
     \end{array}
   \right.
\end{equation}
The two inequalities of (\ref{cond1}) cannot hold simultaneously. The case $wt(j)>3$ can be handled in a  similar manner. We are therefore in position to complete the proofs of (i) and (ii). Indeed, since $|C_j^{\prime}|>1$,  there exists some $j_t\in C_j^{\prime}$ and $j_t\neq j$.
 From (ii) of this lemma, we have $wt(j)=wt(j_t)=2$ and $|C_j^{\prime}|\leq 2$.

If $m$ is odd and $m\geq 3h$, then $B_2^{\prime}\subseteq B_2$.
For any $j=2^{m-3h+i}+2^{m-h}\in B_2^{\prime}$, we have $2^{2h-i}\cdot j\equiv2^{h-i}+2^{m-h}({\rm{mod}}~2^m-1)$ and $m-3h+1\leq h-i<h-1$.
 Note that $j\neq 2^{2h-i}j$ for all $i\in \{1,2,\ldots,4h-m-1\}$ if $m$ is odd. Then $\{j,2^{2h-i}\cdot j\}\in {B_2^{\prime}\choose 2}$. From (iii) of this lemma, $|C_j^{\prime}|=2$.
 When $m$ is even and $m\geq 3h$, we can similarly prove the condition is sufficient.
If $m$ is odd and $m< 3h$, then $B_2^{\prime\prime}\subseteq B_2$. For any $j=2^i+2^{m-h}\in B_2^{\prime\prime}$, we have
$2^{m-h-i}\cdot j\equiv2^{m-2h-i}+2^{m-h}({\rm{mod}}~2^m-1)$ and $0\leq m-2h-i\leq m-2h<h$. Note that $j\neq 2^{m-2h-i}\cdot j$ for all $i\in\{0,1,\ldots,m-2h\}$ if $m$ is odd. Then $\{j,2^{m-2h-i}\cdot j\}\in {B_2^{\prime\prime}\choose 2}.$ From (iii) of this lemma, $|C_j^{\prime}|=2$.
 When $m$ is even and $m< 3h$, we can similarly prove the condition is sufficient.

 Next, let us prove that the condition of (iv) is necessary. We shall only prove it in  the case where $m$ is odd since the case where $m$ is even can be proved similarly.
 If $m\geq 3h$, then $B_2^{\prime}\subseteq B_2$ and $B_2^{\prime\prime}=\emptyset$.   According to the definition of $B_2^{\prime}$, we know that
$$B_2\backslash B_2^{\prime}=\{2^i+2^{m-h}:
i \in \{1,2,\ldots,m-3h\}\}.$$
Assume that there exists some $j=2^i+2^{m-h}\in B_2\backslash B_2^{\prime}$ such that $|C_j^{\prime}|=2.$ Then $2^{m-h-i}\cdot j=2^{m-2h-i}+2^{m-h}\in B_2$. However, $h\leq m-2h-i\leq m-2h-1$ for any $i\in\{1,2,\ldots,m-3h\}$, which contradicts to the fact that $2^{m-h-i}\cdot j\in B_2.$ Similarly, we can prove that this condition is still necessary when $m<3h$.
This completes the proof.
\end{proof}


The following lemmas give some answers to Open Problem \ref{DZop1}.

\begin{lemma}\label{lem20}
Let $s^{\infty}$ be the sequence of {\rm (\ref{st1})} and let $h$ satisfy
\begin{equation}\label{solvepro21}
  m\geq 3h>
  \left\{
    \begin{array}{ll}
      \frac{3(m-1)}{4}, & \hbox{if $m\equiv1 ({\rm{mod~}}4)$;} \\
      \frac{3(m-2)}{4}, & \hbox{if $m\equiv2 ({\rm{mod~}}4)$;} \\
      \frac{3(m-3)}{4}, & \hbox{if $m\equiv3 ({\rm{mod~}}4)$;} \\
      \frac{3(m-4)}{4}, & \hbox{if $m\equiv0 ({\rm{mod~}}4)$.}
    \end{array}
  \right.
\end{equation}
Then
the linear span $\mathbb{L}_s$ of $s^{\infty}$ is given by
\begin{equation*}\label{Ls1}
  \mathbb{L}_s=
  \left\{
    \begin{array}{ll}
      \frac{m(2^{h+2}-2)}{3}-(4h-m)m+1, & \hbox{if both $\gcd(6,m)$ and $h$ are odd;} \\
      \frac{m(2^{h+2}+2)}{3}-(4h-m)m+1, & \hbox{if $\gcd(6,m)$ is odd and $h$ is even;} \\
      \frac{m(2^{h+2}+1)}{3}-(4h-m)m-\frac{m}{2}, & \hbox{if $\gcd(6,m)$ is even and $h$ are odd;} \\
      \frac{m(2^{h+2}+5)}{3}-(4h-m)m-\frac{m}{2}, & \hbox{if $\gcd(6,m)$ is even, $h$ is even, and $m\neq 3h$;} \\
      \frac{3h(2^{h+2}+5)}{3}-3h^2, & \hbox{if $m=3h$ and $h$ is even.}
    \end{array}
  \right.
\end{equation*}
Moreover, we have
\begin{equation*}\label{Ms1.1.2}
  \mathbb{M}_s(x)=\left\{
                    \begin{array}{ll}
                      (x-1)^{\mathbb{N}_2(m)}
                      \prod\limits_{i_1\in B\setminus \{B_2^{\prime}
                      \cup\{2^{m-h}\}\}}m_{\alpha^{-i_1}}(x)&~\\

                       \prod\limits_{{\tiny\begin{array}{c}
                      3\leq 2j+1\leq 2^h-1\\
                      \kappa_{2j+1}^{(h)}=1\end{array}
                      }}m_{\alpha^{-2j-1}}(x),&\hbox{if $h$ is odd;} \\
                      (x-1)^{\mathbb{N}_2(m)}
                      \prod\limits_{i_1\in B\setminus B_2^{\prime}
                      }m_{\alpha^{-i_1}}(x)&~\\
                       \prod\limits_{{\tiny\begin{array}{c}
                      3\leq 2j+1\leq 2^h-1\\
                      \kappa_{2j+1}^{(h)}=1\end{array}
                      }}m_{\alpha^{-2j-1}}(x),& \hbox{if $h$ is even.}
                    \end{array}
                  \right.
\end{equation*}
\end{lemma}
\begin{proof}
According to Lemma \ref{lem15}, if $h$ satisfies the conditions of (\ref{solvepro21}), then we have $C_{i+2^{m-h}}\cap C_j
=\emptyset$ when $i+2^{m-h}\in B\setminus\{0\}$ and odd $j\in A$.
From the definition of $B$, we have
$$B=\{2^{m-h}\}\cup B_2^{\prime}\cup\left\{B\setminus\{B_2^{\prime}\cup
\{2^{m-h}\}\}\right\}.$$
If $h$ is odd, then $\kappa_1^{(h)}=1$ and Eq. (\ref{st1}) can be transformed into
\begin{eqnarray*}
  s_t &=& {\rm{Tr}}\left(\sum_{i=0}^{2^h-1}
  x^{i+2^{m-h}}+\sum_{j=1}^{2^h-1}x^j\right)+ {\rm{Tr}}(1) \\
  ~   &=& {\rm{Tr}}\left(x^{2^{m-h}}+\sum_{i_1\in B\setminus\{B_2^{\prime}\cup
\{2^{m-h}\}\}}
  x^{i_1}+
  \sum_{i_2\in B_2^{\prime}}x^{i_2}
  +\sum_{j=1}^{2^h-1}x^j\right)+ {\rm{Tr}}(1) \\
  ~ &=& {\rm{Tr}}\left(\sum_{i_1\in B\setminus\{B_2^{\prime}\cup
\{2^{m-h}\}\}}
  x^{i_1}+
  \sum\limits_{{\tiny\begin{array}{c}
                      3\leq 2j+1\leq 2^h-1\\
                      \kappa_{2j+1}^{(h)}=1\end{array}
                      }}x^{2j+1}\right)+{\rm{Tr}}(1).
\end{eqnarray*}
Note that $m_{\alpha^{-1}}(x)$ is not a factor of $\mathbb{M}_s(x)$ since $\kappa_1^{(h)}=1.$ Combining Lemma \ref{lem16} and Lemma \ref{lem17}, we deduce that the minimal polynomial of $s^{\infty}$ is
\begin{equation}\label{mini1}
  (x-1)^{\mathbb{N}_2(m)}
                      \prod\limits_{i_1\in B\setminus \{B_2^{\prime}
                      \cup\{2^{m-h}\}\}}m_{\alpha^{-i_1}}(x)
                       \prod\limits_{{\tiny\begin{array}{c}
                      3\leq 2j+1\leq 2^h-1\\
                      \kappa_{2j+1}^{(h)}=1\end{array}
                      }}m_{\alpha^{-2j-1}}(x).
\end{equation}

If $h$ is even, then $\kappa_1^{(h)}=0$ and Eq. (\ref{st1}) can be transformed into
\begin{eqnarray*}
  s_t &=&  {\rm{Tr}}\left(\sum_{i_1\in B\setminus B_2^{\prime}}
  x^{i_1}+
  \sum_{i_2\in B_2^{\prime}}x^{i_2}
  +\sum_{j=1}^{2^h-1}x^j\right)+ {\rm{Tr}}(1) \\
  ~ &=& {\rm{Tr}}\left(\sum_{i_1\in B\setminus B_2^{\prime}}
  x^{i_1}+
  \sum\limits_{{\tiny\begin{array}{c}
                      3\leq 2j+1\leq 2^h-1\\
                      \kappa_{2j+1}^{(h)}=1\end{array}
                      }}x^{2j+1}\right)+{\rm{Tr}}(1).
\end{eqnarray*}
Similarly, one gets that the minimal polynomial of $s^{\infty}$ is
\begin{equation}\label{mini2}
  (x-1)^{\mathbb{N}_2(m)}
                      \prod\limits_{i_1\in B\setminus B_2^{\prime}
                      }m_{\alpha^{-i_1}}(x)
                       \prod\limits_{{\tiny\begin{array}{c}
                      3\leq 2j+1\leq 2^h-1\\
                      \kappa_{2j+1}^{(h)}=1\end{array}
                      }}m_{\alpha^{-2j-1}}(x).
\end{equation}

According to the assumption that $3h\leq m$, if $\gcd(6,m)$ is odd, then $\mathbb{N}_2(m)=1$ and the size of $C_{i+2^{m-h}}$ equals $m$ for any $i\in A$.
When $h$ is odd, from Lemma \ref{DZlem9} and (\ref{mini1}), the linear span $\mathbb{L}_s$ of $s^{\infty}$ is given as follows.
\begin{eqnarray}\label{ls1.1}
  \nonumber\mathbb{L}_s &=& \mathbb{N}_2(m)+\left(\frac{2^h+(-1)^{h-1}}{3}-1+\left|B\setminus \{B_2^{\prime}
                      \cup\{2^{m-h}\}\}\right|\right)\cdot m \\
  \nonumber~ &=& \mathbb{N}_2(m)+\frac{m(2^{h+2}+(-1)^{h-1}
                   -6)}{3}-(4h-m-1-\mathbb{N}_2(m+1))\cdot m\\
  ~ &=& \frac{m(2^{h+2}-2)}{3}-(4h-m)m+1.
\end{eqnarray}
When $h$ is even, from Lemma \ref{DZlem9} and  (\ref{mini2}), the linear span $\mathbb{L}_s$ of $s^{\infty}$ is given as follows.

\begin{eqnarray}\label{ls1.2}
  \nonumber\mathbb{L}_s &=& \mathbb{N}_2(m)+\left(\frac{2^h+(-1)^{h-1}}{3}+\left|B\setminus B_2^{\prime}
                      \right|\right)\cdot m \\
  \nonumber~ &=& \mathbb{N}_2(m)+\frac{m(2^{h+2}+(-1)^{h-1})}{3}-(4h-m-1-\mathbb{N}_2(m+1))\cdot m\\
  ~ &=& \frac{m(2^{h+2}+2)}{3}-(4h-m)m+1.
\end{eqnarray}
If $\gcd(6,m)$ is even, then $\mathbb{N}_2(m)=0$. Using Lemma \ref{lem16} and (\ref{mini1}), if $h$ is odd and $m\neq 3h$, then $2^{\frac{m}{2}-h}+2^{m-h}\in B_2^{\prime}$ and the linear span $\mathbb{L}_s$ of $s^{\infty}$ equals
\begin{eqnarray}\label{ls1.3}
  \nonumber\mathbb{L}_s &=& \mathbb{N}_2(m)+\left(\frac{2^h+(-1)^{h-1}}{3}-1+\left|B\setminus \{B_2^{\prime}
                      \cup\{2^{m-h}\}\}\right|\right)\cdot m-\frac{m}{2} \\
  \nonumber~ &=& \mathbb{N}_2(m)+\frac{m(2^{h+2}+(-1)^{h-1}
                   -6)}{3}-(4h-m-1-\mathbb{N}_2(m+1))\cdot m-\frac{m}{2}\\
  ~ &=& \frac{m(2^{h+2}+1)}{3}-(4h-m)m-\frac{m}{2}.
\end{eqnarray}
Note that if $h$ is odd, then $m\neq 3h$ since $\gcd(6,m)$ is even. If $h$ is even and $m\neq3h$, then we have
\begin{eqnarray}\label{ls1.4}
  \nonumber\mathbb{L}_s &=& \mathbb{N}_2(m)+\left(\frac{2^h+(-1)^{h-1}}{3}+\left|B\setminus B_2^{\prime}
                      \right|\right)\cdot m-\frac{m}{2} \\
  \nonumber~ &=& \mathbb{N}_2(m)+\frac{m(2^{h+2}+(-1)^{h-1})}{3}
  -(4h-m-1-\mathbb{N}_2(m+1))\cdot m-\frac{m}{2}\\
  ~ &=& \frac{m(2^{h+2}+5)}{3}-(4h-m)m-\frac{m}{2}.
\end{eqnarray}
If $m=3h$ and $h$ is even, then $2^{\frac{m}{2}-h}+2^{m-h}\notin B_2^{\prime}$  and
\begin{eqnarray}\label{ls1.5}
  \nonumber\mathbb{L}_s &=& \mathbb{N}_2(m)+\left(\frac{2^h+(-1)^{h-1}}{3}+\left|B\setminus B_2^{\prime}
                      \right|\right)\cdot m \\
  \nonumber~ &=& \mathbb{N}_2(m)+\frac{m(2^{h+2}+(-1)^{h-1})}{3}
  -(4h-m-1-\mathbb{N}_2(m+1))\cdot m\\
  ~ &=& \frac{3h(2^{h+2}+5)}{3}-3h^2.
\end{eqnarray}
Combining Eqs. (\ref{ls1.1}), (\ref{ls1.2}), (\ref{ls1.3}), (\ref{ls1.4}), and (\ref{ls1.5}), we obtain the desired results.
\end{proof}
The next theorem provides interesting information on the cyclic code $\C_s$ when $h$ satisfies
{\rm{(\ref{solvepro21})}}.
\begin{theorem}\label{thm23}
	Let $s^{\infty}$ be the sequence of {\rm (\ref{st1})} and let $h$ satisfy {\rm{(\ref{solvepro21})}}.
	The binary code $\C_s$ defined by the sequence of {\rm{(\ref{st1})}} has parameters $[2^m-1,2^m-1-\mathbb{L}_s,d]$ and generator polynomial $\mathbb{M}_s(x)$, where $\mathbb{L}_s$ and $\mathbb{M}_s(x)$ are given in Lemma \ref{lem20} and the minimum weight $d$ has the following bounds:
	\begin{equation*}
		d\geq\left\{
		\begin{array}{ll}
			2^{m-3h+1}, & \hbox{if $h$ is odd;} \\
			2^{m-3h+1}+1, & \hbox{if $h$ is even and $m$ is even;} \\
			2^{m-3h+1}+2, & \hbox{if $h$ is even and $m$ is odd.}
		\end{array}
		\right.
	\end{equation*}
\end{theorem}
\begin{proof}
The desired conclusion on the dimension of $\C_s$ follows from Lemma \ref{lem20}. When $h$ is even, from (\ref{Ms1.1.2}) and the definition of $B_2^{\prime}$, the designed distance $\delta$ of $\C_s$ is greater than $2^{m-3h+1}-1$. According to the BCH bound, we have $d\geq \delta\geq 2^{m-3h+1}$. Similarly, we can prove the remaining cases.
\end{proof}
\begin{example}
Let $(m,h)=(7,2)$ and $\alpha$ be a generator of $\F_{2^m}$ with $\alpha^7+\alpha+1=0$. Then the generator polynomial of the cyclic code $\C_s$ is $\mathbb{M}_s(x)=x^{36}+x^{28}+x^{27}+x^{23}
+x^{21}+x^{20}+x^{18}+x^{13}+x^{12}+x^9+x^7
+x^6+x^5+1$ and $\C_s$ is a $[127,91,8]$ binary cyclic code. It is worth mentioning that this numerical example is derived from \cite[Example 8]{DZ}, but $m$ and $h$ of this example do not satisfy the restrictions on $m$ and $h$ in \cite[Theorem 23]{DZ}. At this point $\deg(\mathbb{M}_s(x))=
\frac{m(2^h+2+(-1)^{h-1})+3}{3}=36$ is just a coincidence. The correct degree of $\mathbb{M}_s(x)$  can be derived from Theorem \ref{thm23}:
$$\deg(\mathbb{M}_s(x))=
\frac{m(2^{h+2}+2)}{3}-(4h-m)m+1=36.$$
\end{example}
\begin{lemma}\label{lem21}
Let $s^{\infty}$ be the sequence of {\rm (\ref{st1})} and let $h$ satisfy $m< 3h$ and $5h<2m+3$.
Then the linear span $\mathbb{L}_s$ of $s^{\infty}$ is given by
\begin{equation*}\label{Ls1}
  \mathbb{L}_s=
  \left\{
    \begin{array}{ll}
      \frac{m(2^{h+2}-8)}{3}-(m-2h)m+1, & \hbox{if $\gcd(6,m)=1$ and $h$ is odd;} \\
      \frac{m(2^{h+2}-4)}{3}-(m-2h)m+1, & \hbox{if $\gcd(6,m)=1$ and $h$ is even;} \\
      \frac{m(2^{h+2}-5)}{3}-(\frac{3m}{2}-2h)m, & \hbox{if $\gcd(6,m)=2$ and $h$ is odd;} \\
      \frac{m(2^{h+2}-1)}{3}-(\frac{3m}{2}-2h)m, & \hbox{if $\gcd(6,m)=2$ and $h$ is even;} \\
      \frac{m(2^{h+2}-8)}{3}-(\frac{4m}{3}-2h)m+1, & \hbox{if $\gcd(6,m)=3$ and $h$ is odd;} \\
      \frac{m(2^{h+2}-4)}{3}-(\frac{4m}{3}-2h)m+1, & \hbox{if $\gcd(6,m)=3$ and $h$ is even;} \\
      \frac{m(2^{h+2}-5)}{3}-(\frac{13m}{6}-2h)m, & \hbox{if $\gcd(6,m)=6$ and $h$ is odd;} \\
      \frac{m(2^{h+2}-1)}{3}-(\frac{13m}{6}-2h)m, & \hbox{if $\gcd(6,m)=6$ and $h$ is even.}
    \end{array}
  \right.
\end{equation*}
Moreover, we have
\begin{equation*}\label{Ms1.1.2}
  \mathbb{M}_s(x)=\left\{
                    \begin{array}{ll}
                      (x-1)^{\mathbb{N}_2(m)}
                      \prod\limits_{i_1\in B\setminus \{B_2^{\prime\prime}
                      \cup\{2^{m-h}\}\}}m_{\alpha^{-i_1}}(x)&~\\
                       \prod\limits_{{\tiny\begin{array}{c}
                      3\leq 2j+1\leq 2^h-1\\
                      \kappa_{2j+1}^{(h)}=1\end{array}
                      }}m_{\alpha^{-2j-1}}(x),& \hbox{if $h$ is odd;} \\
                      (x-1)^{\mathbb{N}_2(m)}
                      \prod\limits_{i_1\in B\setminus B_2^{\prime\prime}
                      }m_{\alpha^{-i_1}}(x)&~\\
                       \prod\limits_{{\tiny\begin{array}{c}
                      3\leq 2j+1\leq 2^h-1\\
                      \kappa_{2j+1}^{(h)}=1\end{array}
                      }}m_{\alpha^{-2j-1}}(x),& \hbox{if $h$ is even.}
                    \end{array}
                  \right.
\end{equation*}
\end{lemma}
\begin{proof}
The proof of this lemma is similar to the one given in Lemma \ref{lem20}.
According to Lemma \ref{lem15}, if $h$ satisfies $m< 3h$ and $5h<2m+3$, then we have $C_{i+2^{m-h}}\cap C_j
=\emptyset$ when $i+2^{m-h}\in B\setminus\{0\}$ and odd $j\in A$.
From the definition of $B$, we have
$$B=\{2^{m-h}\}\cup B_2^{\prime\prime}\cup\left\{B\setminus\{B_2^{\prime\prime}\cup
\{2^{m-h}\}\}\right\}.$$
Combining Eq. (\ref{st1}), Lemma \ref{DZlem9} and Lemma \ref{192021}, we deduce that the minimal polynomial of $s^{\infty}$ is
\begin{equation*}\label{Ms1.1.2}
  \mathbb{M}_s(x)=\left\{
                    \begin{array}{ll}
                      (x-1)^{\mathbb{N}_2(m)}
                      \prod\limits_{i_1\in B\setminus \{B_2^{\prime\prime}
                      \cup\{2^{m-h}\}\}}m_{\alpha^{-i_1}}(x)&~\\
                       \prod\limits_{{\tiny\begin{array}{c}
                      3\leq 2j+1\leq 2^h-1\\
                      \kappa_{2j+1}^{(h)}=1\end{array}
                      }}m_{\alpha^{-2j-1}}(x),& \hbox{if $h$ is odd;} \\
                      (x-1)^{\mathbb{N}_2(m)}
                      \prod\limits_{i_1\in B\setminus B_2^{\prime\prime}
                      }m_{\alpha^{-i_1}}(x)&~\\
                       \prod\limits_{{\tiny\begin{array}{c}
                      3\leq 2j+1\leq 2^h-1\\
                      \kappa_{2j+1}^{(h)}=1\end{array}
                      }}m_{\alpha^{-2j-1}}(x),& \hbox{if $h$ is even.}
                    \end{array}
                  \right.
\end{equation*}
If $\gcd(6,m)=1$, then $\mathbb{N}_2(m)=1$ and
$|C_{i+2^{m-h}}|=m$ for any $i\in A$. If $h$ is odd, then the linear span $\mathbb{L}_s$ of $S^{\infty}$ equals
\begin{eqnarray}\label{ls1.6}
  \nonumber\mathbb{L}_s &=& \mathbb{N}_2(m)+\left(\frac{2^h+(-1)^{h-1}}{3}-1+\left|B\setminus \{B_2^{\prime\prime}\cup\{2^{m-h}\}\}\right|\right)\cdot m\\
             \nonumber~&=&          \mathbb{N}_2(m)+\left(\frac{2^h+(-1)^{h-1}}{3}-1+2^h-1-(m-2h+1-\mathbb{N}_2(m+1))\right)\cdot m \\
  ~ &=&\frac{m(2^{h+2}-8)}{3}-(m-2h)m+1.
\end{eqnarray}
If $h$ is even, then the linear span $\mathbb{L}_s$ of $S^{\infty}$ equals
\begin{eqnarray}\label{ls1.7}
  \nonumber\mathbb{L}_s &=& \mathbb{N}_2(m)+\left(\frac{2^h+(-1)^{h-1}}{3}+\left|B\setminus B_2^{\prime\prime}\right|\right)\cdot m\\
             \nonumber~&=&          \mathbb{N}_2(m)+\left(\frac{2^h+(-1)^{h-1}}{3}+2^h-(m-2h+1-\mathbb{N}_2(m+1))\right)\cdot m \\
  ~ &=& \frac{m(2^{h+2}-4)}{3}-(m-2h)m+1.
\end{eqnarray}
If $\gcd(6,m)=2$, then we get that $\mathbb{N}_2(m)=0$, $\left|C_{2^{{\frac{m}{2}}-h}+2^{m-h}}\right|=\frac{m}{2}$ and $C_{2^{{\frac{m}{2}}-h}+2^{m-h}}\notin
B_2^{\prime\prime}$. If $h$ is odd, then
\begin{small}\begin{eqnarray}\label{ls1.8}
  \nonumber\mathbb{L}_s &=& \mathbb{N}_2(m)+\left(\frac{2^h+(-1)^{h-1}}{3}-1+\left|B\setminus \{B_2^{\prime\prime}\cup\{2^{m-h}\}\}\right|\right)\cdot m-\frac{m}{2}\\
             \nonumber~&=&          \mathbb{N}_2(m)+\left(\frac{2^h+(-1)^{h-1}}{3}-1+2^h-1-(m-2h+1-\mathbb{N}_2(m+1))-\frac{1}{2}\right)\cdot m \\
  ~ &=&\frac{m(2^{h+2}-5)}{3}-(\frac{3m}{2}-2h)m.
\end{eqnarray}
\end{small}
If $h$ is even, then
\begin{eqnarray}\label{ls1.9}
  \nonumber\mathbb{L}_s &=& \mathbb{N}_2(m)+\left(\frac{2^h+(-1)^{h-1}}{3}+\left|B\setminus B_2^{\prime\prime}\right|\right)\cdot m-\frac{m}{2}\\
  \nonumber~&=&          \mathbb{N}_2(m)+\left(\frac{2^h+(-1)^{h-1}}{3}+2^h-(m-2h+1-\mathbb{N}_2(m+1))\right)\cdot m
  -\frac{m}{2} \\
  ~ &=& \frac{m(2^{h+2}-1)}{3}-(\frac{3m}{2}-2h)m.
\end{eqnarray}
If $\gcd(6,m)=3$, then $\mathbb{N}_2(m)=1$ and $\left|C_{{2^{\frac{m}{3}}-h}+2^{{\frac{2m}{3}}-h}+2^{m-h}}\right|=\frac{m}{3}$. If $h$ is odd, then
\begin{small}
\begin{eqnarray}\label{ls1.10}
  \nonumber\mathbb{L}_s &=& \mathbb{N}_2(m)+\left(\frac{2^h+(-1)^{h-1}}{3}-1+\left|B\setminus \{B_2^{\prime\prime}\cup\{2^{m-h}\}\}\right|\right)\cdot m-\frac{m}{3}\\
             \nonumber~&=&          \mathbb{N}_2(m)+\left(\frac{2^h+(-1)^{h-1}}{3}-1+2^h-1-(m-2h+1-\mathbb{N}_2(m+1))\right)\cdot m-\frac{m}{3} \\
  ~ &=&\frac{m(2^{h+2}-8)}{3}-(\frac{4m}{3}-2h)m+1.
\end{eqnarray}
\end{small}
If $h$ is even, then
\begin{eqnarray}\label{ls1.11}
  \nonumber\mathbb{L}_s &=& \mathbb{N}_2(m)+\left(\frac{2^h+(-1)^{h-1}}{3}+\left|B\setminus B_2^{\prime\prime}\right|\right)\cdot m-\frac{m}{3}\\
  \nonumber~&=&          \mathbb{N}_2(m)+\left(\frac{2^h+(-1)^{h-1}}{3}+2^h-(m-2h+1-\mathbb{N}_2(m+1))\right)\cdot m-\frac{m}{3} \\
  ~ &=& \frac{m(2^{h+2}-4)}{3}-(\frac{4m}{3}-2h)m+1.
\end{eqnarray}
If $\gcd(6,m)=6$, then $\mathbb{N}_2(m)=0$, $\left|C_{2^{{\frac{m}{2}}-h}+2^{m-h}}\right|=\frac{m}{2}$,  $\left|C_{{2^{{\frac{m}{3}}-h}}+2^{{\frac{2m}{3}}-h}+2^{m-h}}\right|=\frac{m}{3}$ and $C_{2^{{\frac{m}{2}}-h}+2^{m-h}}\notin B_2^{\prime\prime}$. If $h$ is odd, then
\begin{small}
	\begin{eqnarray}\label{ls1.12}
  \nonumber\mathbb{L}_s &=& \mathbb{N}_2(m)+\left(\frac{2^h+(-1)^{h-1}}{3}-1+\left|B\setminus \{B_2^{\prime\prime}\cup\{2^{m-h}\}\}\right|\right)\cdot m-\frac{m}{2}-\frac{m}{3}\\
             \nonumber~&=&          \mathbb{N}_2(m)+\left(\frac{2^h+(-1)^{h-1}}{3}-1+2^h-1-(m-2h+1-\mathbb{N}_2(m+1))\right)\cdot m-\frac{7m}{6} \\
  ~ &=&\frac{m(2^{h+2}-5)}{3}-(\frac{13m}{6}-2h)m.
\end{eqnarray}
\end{small}
If $h$ is even, then
\begin{eqnarray}\label{ls1.13}
  \nonumber\mathbb{L}_s &=& \mathbb{N}_2(m)+\left(\frac{2^h+(-1)^{h-1}}{3}+\left|B\setminus B_2^{\prime\prime}\right|\right)\cdot m-\frac{m}{2}-\frac{m}{3}\\
  \nonumber~&=&          \mathbb{N}_2(m)+\left(\frac{2^h+(-1)^{h-1}}{3}+2^h-(m-2h+1-\mathbb{N}_2(m+1))\right)\cdot m-\frac{7m}{6} \\
  ~ &=& \frac{m(2^{h+2}-1)}{3}-(\frac{13m}{6}-2h)m.
\end{eqnarray}
Combining Eqs. (\ref{ls1.6}), (\ref{ls1.7}), (\ref{ls1.8}), (\ref{ls1.9}), (\ref{ls1.10}),
(\ref{ls1.11}),
(\ref{ls1.12}),
 and (\ref{ls1.13}), we obtain the desired results.
\end{proof}
The next theorem also provides interesting information on the cyclic code $\C_s$ when $h$ satisfies
$m< 3h$ and $5h<2m+3$.
\begin{theorem}\label{thm24}
Let $h$ satisfy
$m< 3h$ and $5h<2m+3$.
The binary code $\C_s$ defined by the sequence of {\rm{(\ref{st1})}} has parameters $[2^m-1,2^m-1-\mathbb{L}_s,d]$ and generator polynomial $\mathbb{M}_s(x)$, where $\mathbb{L}_s$ and $\mathbb{M}_s(x)$ are given in Lemma \ref{lem21} and the minimum  Hamming distance $d$  is  bounded as follows
\begin{equation*}
  d\geq 2^h-2^{m-2h}.
\end{equation*}
\end{theorem}
\begin{proof}
The desired conclusions on the dimension of $\C_s$ follow from Lemma \ref{lem20}. It is not difficult to verify that $\{1+2^{m-2h}+2^{m-h},2+2^{m-2h}+2^{m-h},\ldots, 2^h-1+2^{m-h}\}$ is a subset of the defining set of $\C_s$, then the minimum distance $d\geq 2^h-2^{m-2h},$  according to the BCH bound.
\end{proof}
\begin{example}
Let $(m,h)=(5,2)$ and $\alpha$ be a generator of $\F_{2^m}$ with $\alpha^5+\alpha^2+1=0$. Then the generator polynomial of the code $\C_s$ is $\mathbb{M}_s(x)=x^{16}+x^{14}+x^{10}+x^9
+x^8+x^7
+x^5+x^4+x^3+x^2+x+1
$ and $\C_s$ is a $[31,15,8]$ binary cyclic code. Its dual $\C_s^{\perp}$ is a $[31,16,7]$ cyclic code. Both $\C_s$ and $\C_s^{\perp}$ are optimal according to the Database. It is worth mentioning that this numerical example is derived from \cite[Example 7]{DZ}.
\end{example}
\begin{remark}
If $h$ satisfies
$$2m+3\leq 5h\leq
  \left\{
    \begin{array}{ll}
      1\leq h\leq\frac{m-1}{2}, & \hbox{if $m\equiv1 ({\rm{mod~}}4)$,} \\
      1\leq h\leq\frac{m-2}{2}, & \hbox{if $m\equiv2 ({\rm{mod~}}4)$,} \\
      1\leq h\leq\frac{m-3}{2}, & \hbox{if $m\equiv3 ({\rm{mod~}}4)$,} \\
      1\leq h\leq\frac{m-4}{2}, & \hbox{if $m\equiv0 ({\rm{mod~}}4)$,}
    \end{array}
  \right.
$$
then there exist $i\in B$ and odd $j \in A\setminus\{1\}$ such that $C_{i+2^{m-h}}\cap C_j\neq \emptyset$. In this case, the linear span of $\mathbb{L}_s$ hard to determine. The readers are cordially invited to find a strategy to compute it.
\end{remark}
\subsection{Binary cyclic codes from the Bracken-Leander function $F_8(x)$}\label{Bracken-Leander}
Recall that the binary Bracken-Leander function $F_8(x)=x^{2^{2h}+2^{h}+1}$ has $\delta_{F_8}=4$ if  $m=4h$ and $h$ is odd. Let us observe that
\begin{equation*}
  {\rm{Tr}}(F_8(x+1)) = {\rm{Tr}}\left((x+1)^{2^{2h}+2^h+1}\right) = {\rm{Tr}}\left(x+x^
  {1+2^{2h}}+x^{1+2^h+2^{2h}}\right).
\end{equation*}
The sequence $s^{\infty}$ of (\ref{infty}) defined by the monomial $F_8(x)$ is then given by
\begin{equation}\label{st36}
  s_t={\rm{Tr}}\left(\alpha^t+(\alpha^t)^
  {1+2^{2h}}+(\alpha^t)^{1+2^h+2^{2h}}
  \right)
\end{equation} for all $t\geq 0$.
\begin{lemma}\label{lem27}
Let $s^{\infty}$ be the sequence of {\rm{(\ref{st36})}}. Then the linear span $\mathbb{L}_s$ of $s^{\infty}$ is $\frac{5m}{2}$. Moreover, we have
$$\mathbb{M}_s(x)=m_{\alpha^{-1}}(x)m_{\alpha^{-1-2^{2h}}}(x)m_{\alpha^{-1-2^{h}-2^{2h}}}(x).$$
\end{lemma}
\begin{proof}
It is easy to check that for any $i,j \in \{1,1+2^{2h},1+2^h+2^{2h}\}$, $C_i\neq C_j$ if $i\neq j$. Since
$$C_{1+2^{2h}}=\{1+2^{2h},2+2^{2h+1},\cdots,2^{2h-1}+2^{4h-1}\}$$
and $\gcd(1+2^h+2^{2h},2^{4h}-1)=1$, then $|C_1|=|C_{1+2^h+2^{2h}}|=m$ and $|C_{1+2^h}|=\frac{m}{2}.$ Then the desired results follow from (\ref{st36}).
\end{proof}
The following theorem also provides interesting information on the cyclic code $\C_s$ from the Bracken-Leander function $F_8$.
\begin{theorem}\label{thm24}
The binary code $\C_s$ defined by the sequence of {\rm{(\ref{st36})}} has parameters $[2^m-1,2^m-1-\frac{5m}{2},3]$.
\end{theorem}
\begin{proof}
The dimension of $\C_s$ follows from Lemma \ref{lem27}. According to the BCH bound, we have $d(\C_s)\geq 3$. Since $4|m$, we have $3|2^m-1$. There exist $\alpha^{\frac{2^m-1}{3}},
\alpha^{\frac{2(2^m-1)}{3}}$ and
$1$ such that
\begin{equation*}
  \left\{
     \begin{array}{ll}
       \alpha^{\frac{2^m-1}{3}}+
\alpha^{\frac{2(2^m-1)}{3}}+1 = 0, \\
       (\alpha^{\frac{2^m-1}{3}})^{1+2^{2h}}+
(\alpha^{\frac{2(2^m-1)}{3}})^{1+2^{2h}}+1 = 0, \\
       (\alpha^{\frac{2^m-1}{3}})^{1+2^h+2^{2h}}+
(\alpha^{\frac{2(2^m-1)}{3}})^{1+2^h+2^{2h}}+1 = 0.
     \end{array}
   \right.
\end{equation*}
since $\gcd(3,2^{2h}+1)=1$ and $\gcd(3,2^{2h}+2^h+1)=1$ when $h$ is odd. Therefore, the minimum  Hamming distance of $\C_s$ equals $3$.
\end{proof}
\begin{example}
Let $m=4$ and $\alpha$ be a generator of $\F_{2^m}$ with $\alpha^4+\alpha+1=0$. Then the generator polynomial of the code $\C_s$ is given by $\mathbb{M}_s(x)=x^{10}+x^{5}+1
$ and $\C_s$ is a $[15,5,3]$ binary cyclic code.
\end{example}
 In \cite{DSIAM}, Ding considered $q$-ary cyclic codes from the monomial $f(x)=x^{\frac{q^h-1}{q-1}}$.  When $h=3$,  he gave the following result.
\begin{lemma}\label{lem29}\cite[Corollary 5.15]{DSIAM}
Let $m$ be a positive integer with
 $$3\leq \left\{
         \begin{array}{ll}
           \frac{{m}-1}{2} & \hbox{if $m$ is odd,} \\
           \frac{{m}}{2} & \hbox{if $m$ is even.}
         \end{array}
       \right.
$$
 Let $s^{\infty}$ be the sequence of {\rm{(\ref{infty})}}, where $f(x)=x^{q^2+q+1}$ over $\F_q$. Then the code $\C_s$ has parameters $[n,n-\mathbb{L}_s,d]$ and generator polynomial $\mathbb{M}_s(x)$ given by
$$
\mathbb{M}_s(x)=(x-1)^{\mathbb{N}_p(m)}m_{\alpha^{-1}}(x)m_{\alpha^{-1-q}}(x)m_{\alpha^{-1-q^2}}(x)
m_{\alpha^{-1-q-q^2}}(x)
$$
if $p\neq 3$, and
$$
\mathbb{M}_s(x)=(x-1)^{\mathbb{N}_p(m)}m_{\alpha^{-1-q}}(x)m_{\alpha^{-1-q^2}}(x)
m_{\alpha^{-1-q-q^2}}(x)
$$
if $p= 3$,
where $$\mathbb{L}_s=\left\{
                       \begin{array}{ll}
                         4m+\mathbb{N}_p(m) & \hbox{if $p \neq 3$,} \\
                         3m+\mathbb{N}_p(m) & \hbox{if $p =3$.}
                       \end{array}
                     \right.
$$
In addition,
$$
\left\{
  \begin{array}{ll}
    3\leq d\leq 8 & \hbox{if $p=3$ and $\mathbb{N}_p(m)=1$,} \\
    3\leq d\leq 6 & \hbox{if $p=3$ and $\mathbb{N}_p(m)=0$,} \\
    3\leq d\leq 8 & \hbox{if $p>3$.}
  \end{array}
\right.
$$
\end{lemma}
The following is an attractive open problem proposed in \cite{DSIAM} by Ding.
\begin{open}\label{open30}
For the code $\C_s$ of Lemma \ref{lem29}, do the following lower bounds hold?
\begin{equation*}
  d\geq\left\{
         \begin{array}{ll}
           5, & \hbox{where $p=3$ and $\mathbb{N}_p(m)=1$;} \\
           4, & \hbox{where $p=3$ and $\mathbb{N}_p(m)=0$;} \\
           6, & \hbox{where $p>3$ and $\mathbb{N}_p(m)=1$;} \\
           5, & \hbox{where $p>3$ and $\mathbb{N}_p(m)=0$.}
         \end{array}
       \right.
\end{equation*}
\end{open}
However, in  \cite{DSIAM} the  case $p=2$ was  missed. Next, we shall complete handling  that  remaining case and  also provide some answers to Open Problem \ref{open30}. To this end, we need the following fundamental result on elementary number theory.
\begin{lemma}\label{lem31}
Let $h\geq 1$ and let $a>1$ be an integer. Then
\begin{equation*}
  \gcd(a^l+1,a^h-1)=\left\{
                      \begin{array}{ll}
                        1, & \hbox{if $\frac{h}{\gcd(l,h)}$ is odd and $a$ is even;} \\
                        2, & \hbox{if $\frac{h}{\gcd(l,h)}$ is odd and $a$ is odd;} \\
                        a^{\gcd(l,h)}+1, & \hbox{if $\frac{h}{\gcd(l,h)}$ is even.}
                      \end{array}
                    \right.
\end{equation*}
\end{lemma}
\begin{theorem}\label{thm33}
 Let $q=p^s$ and let the code $\C_s$ be defined as the same as Lemma \ref{lem29}. Then $\C_s$ has parameters $[n,n-\mathbb{L}_s,d]$ and generator polynomial $M_s(x)$ given by
$$
\mathbb{M}_s(x)=(x-1)^{\mathbb{N}_p(m)}m_{\alpha^{-1}}(x)m_{\alpha^{-1-q}}(x)m_{\alpha^{-1-q^2}}(x)
m_{\alpha^{-1-q-q^2}}(x)
$$
if $p\notin\{2,3\}$,
$$
\mathbb{M}_s(x)=(x-1)^{\mathbb{N}_p(m)}m_{\alpha^{-1}}(x)m_{\alpha^{-1-q^2}}(x)
m_{\alpha^{-1-q-q^2}}(x)
$$
if $p=2$,
 and
$$
\mathbb{M}_s(x)=(x-1)^{\mathbb{N}_p(m)}m_{\alpha^{-1-q}}(x)m_{\alpha^{-1-q^2}}(x)
m_{\alpha^{-1-q-q^2}}(x)
$$
if $p= 3$,
where $$\mathbb{L}_s=\left\{
                       \begin{array}{ll}
                         4m+\mathbb{N}_p(m) & \hbox{if $p \notin \{2,3\}$,} \\
                         3m+\mathbb{N}_p(m) & \hbox{if $p \in \{2,3\}$.}
                       \end{array}
                     \right.
$$
In addition,
$$
\left\{
  \begin{array}{ll}
    d=3 & \hbox{if $p=2$, $s$ is odd, and $m$ is even;} \\
    3\leq d\leq 5 & \hbox{if $p=2$, $s\equiv2({\rm{mod~}}4)$, and $m$ is even;} \\
    3\leq d\leq 4 & \hbox{if $p=2$, $s\equiv0({\rm{mod~}}4)$, and $m$ is even;} \\
    3\leq d\leq 5 & \hbox{if $p=2$, $s$ is even, and $m$ is odd;} \\
    3\leq d\leq 5 & \hbox{if $p=2$, $s$ is odd, $m$ is odd, and}\\
     ~&\hbox{$q-1$ has a proper factor $\eta$ which is greater than $4$;} \\
    3\leq d\leq 8 & \hbox{if $p=2$, $s$ is odd, $m$ is odd, and}\\
     ~&\hbox{$q-1$ has no proper factor which is greater than $4$;} \\
    d=3 & \hbox{if $p=3$, $\mathbb{N}_p(m)=0$, and} \\
    ~&\hbox{$q-1$ has a proper factor $\eta$ which is greater than $3$;}\\
    3\leq d\leq 6 & \hbox{if $p=3$, $\mathbb{N}_p(m)=0$, and} \\
    ~&\hbox{$q-1$ has no proper factor which is greater than $3$;}\\
    3\leq d\leq 4 & \hbox{if $p=3$, $\mathbb{N}_p(m)=1$, and} \\
    ~&\hbox{$q-1$ has a proper factor $\eta$ which is greater than $3$;}\\
    3\leq d\leq 8 & \hbox{if $p=3$, $\mathbb{N}_p(m)=1$, and} \\
    ~&\hbox{$q-1$ has no proper factor which is greater than $3$;}\\
    3\leq d\leq 4 & \hbox{if $p=5$ and $\mathbb{N}_p(m)=0$;}\\
    3\leq d\leq 5 & \hbox{if $p=5$, $\mathbb{N}_p(m)=1$, and}\\
     ~&\hbox{$q-1$ has a proper factor $\eta$ which is greater than $4$;}\\
    3\leq d\leq 4 & \hbox{if $p=5$ and $\mathbb{N}_p(m)=1$, and}\\
    ~&\hbox{$q-1$ has no proper factor which is greater than $4$;}\\
    3\leq d\leq 4 & \hbox{if $p>5$ and $\mathbb{N}_p(m)=0$;}\\
    3\leq d\leq 5 & \hbox{if $p>5$ and $\mathbb{N}_p(m)=1$.}\\
  \end{array}
\right.
$$
\end{theorem}
\begin{proof}
Let us observe that
\begin{small}\begin{eqnarray*}
  {\rm{Tr}}_{r/q}\left(f(x+1)\right)  &=&
  {\rm{Tr}}_{r/q}\left(x^{q^{2}+q+1}+(x^{q^2+q}+x^{q+1})+x^{q^2+1}+(x^{q^2}+x^q+x)+1\right)  \\
  ~ &=& \left\{
          \begin{array}{ll}
            {\rm{Tr}}_{r/q}\left(x^{q^{2}+q+1}+(x^{q^2+q}+x^{q+1})+(x^{q^2}+x^q+x)+1\right), & \hbox{if $p\geq 5$;} \\
            {\rm{Tr}}_{r/q}\left(x^{q^{2}+q+1}+x^{q^2+1}+x+1\right), & \hbox{if $p=2$;} \\
            {\rm{Tr}}_{r/q}\left(x^{q^{2}+q+1}+x^{q^2+1}-x^{q+1}+1\right), & \hbox{if $p=3$.}
          \end{array}
        \right.
\end{eqnarray*}
\end{small}
Then the minimal polynomial $\mathbb{M}_s(x)$ is
$$(x-1)^{\mathbb{N}_p(m)}m_{\alpha^{-1}}(x)m_{\alpha^{-1-q^2}}(x)
m_{\alpha^{-1-q-q^2}}(x)$$
 when $p=2$.

We discuss the upper and lower bounds of the minimum distance of $\C_s$ according to the characteristic of $\F_q$.
\begin{enumerate}
  \item[(i)]$p=2$

  In this case, the parity-check matrix of $\C_s$ is
\begin{equation*}
  \left(
    \begin{array}{cccc}
      \alpha & \alpha^2 & \cdots & \alpha^{n-1} \\
      \alpha^{1+q^2} & (\alpha^2)^{1+q^2} & \cdots & (\alpha^{n-1})^{1+q^2} \\
      \alpha^{1+q+q^2} & (\alpha^2)^{1+q+q^2} & \cdots & (\alpha^{n-1})^{1+q+q^2}
    \end{array}
  \right)
\end{equation*}
if $\mathbb{N}_p(m)=0$ and is
\begin{equation*}
  \left(
    \begin{array}{cccc}
      1 & 1 & \cdots & 1 \\
      \alpha & \alpha^2 & \cdots & \alpha^{n-1} \\
      \alpha^{1+q^2} & (\alpha^2)^{1+q^2} & \cdots & (\alpha^{n-1})^{1+q^2} \\
      \alpha^{1+q+q^2} & (\alpha^2)^{1+q+q^2} & \cdots & (\alpha^{n-1})^{1+q+q^2}
    \end{array}
  \right)
\end{equation*}if $\mathbb{N}_p(m)=1$.
  According to the BCH bound and the Sphere Packing bound, we have
$$
\left\{
  \begin{array}{ll}
    3\leq d\leq 6 & \hbox{if  $m$ is even,} \\
    3\leq d\leq 8 & \hbox{if  $m$ is odd.} \\
  \end{array}
\right.
$$
If $s$ is odd and $m$ is even, then $3|q^m-1$. According to Lemma \ref{lem31}, one gets $\gcd(3,1+q^2)=1$ and $\gcd(3,1+q+q^2)=1$. Let $\beta=\alpha^{\frac{q^m-1}{3}}\in \F_{q^m}^*$. Then
  \begin{equation*}
    ~\left\{
       \begin{array}{ll}
         \beta+\beta^2+1=0, \\
         \beta^{1+q^2}+(\beta^2)^{1+q^2}+1=0,  \\
         \beta^{1+q+q^2}+(\beta^2)^{1+q+q^2}+1=0.
       \end{array}
     \right.
  \end{equation*}
  Therefore, the minimum  Hamming distance of $\C_s$ equals $3$ when $s$ is odd and $m$ is even.

  If $s$ is even and $m$ is even, then $5|q^m-1$. According to Lemma \ref{lem31}, one gets
  $\gcd(5,q^2+1)=1$ and $\gcd(5,1+q+q^2)=1$. Let $\gamma=\alpha^{\frac{q^m-1}{5}}$. There exist $\gamma,\gamma^2,\gamma^3,\gamma^4$ and $1\in \F_q$ such that
  \begin{equation*}
    \left\{
      \begin{array}{ll}
        \gamma+\gamma^2+\gamma^3+\gamma^4+1=0, \\
         \gamma^{1+q^2}+(\gamma^2)^{1+q^2}+(\gamma^3)^{1+q^2}
         +(\gamma^4)^{1+q^2}+1=0,  \\
         \gamma^{1+q+q^2}+(\gamma^2)^{1+q+q^2}+(\gamma^3)^{1+q+q^2}
         +(\gamma^4)^{1+q+q^2}+1=0.
      \end{array}
    \right.
  \end{equation*} Therefore, $3\leq d\leq 5$.
Moreover, if $s\equiv0({\rm{mod~}4})$ and $m$ is even, then $5|q-1$ and $\gamma\in\F_q^*$.
Let $g_1(x):=(x+\gamma)(x+\gamma^2)(x+\gamma^3)=x^3+(\gamma+\gamma^2+\gamma^3)x^2+
(1+\gamma+\gamma^2)\gamma^3x+\gamma$. Then
\begin{equation*}
  \left\{
    \begin{array}{ll}
      g_1(\gamma)=0, \\
      g_1(\gamma^{1+q^2})=g_1(\gamma^2)=0, \\
      g_1(\gamma^{1+q+q^2})=g_1(\gamma^3)=0.
    \end{array}
  \right.
\end{equation*}
Then $3\leq d\leq 4$ since $g_1(x)$ is a tetranomial over $\F_q$.

  If $m$ is odd and $s$ is even, then $5|q-1$ and the all-one vector is a codeword of $\C_s$. Recall that $\gamma=\alpha^{\frac{q^m-1}{5}}$. Let $\sigma_i(x_1,x_2,\ldots,x_l)$ denote the elementary symmetric polynomial of degree $i$ in $l$ variables  $x_1, x_2, \ldots,x_l$. Let
  \begin{eqnarray*}
    g_2(x): &=& (x+1)(x+\gamma)(x+\gamma^2)(x+\gamma^3) \\
    ~ &=& x^4+\sigma_1(1,\gamma,\gamma^2,\gamma^3)x^3
  +\sigma_2(1,\gamma,\gamma^2,\gamma^3)x^2+\sigma_3(1,\gamma,\gamma^2,\gamma^3)x\\
  ~&~&
  +\sigma_4(1,\gamma,\gamma^2,\gamma^3).
  \end{eqnarray*}
  It is easy to verify that $g_2(x)\in \F_q[x]$ and $\sigma_i(1,\gamma,\gamma^2,\gamma^3)\neq 0$ for all $i\in\{1,2,3,4\}.$ Hence, we obtain
  \begin{equation*}
    \left\{
      \begin{array}{ll}
        g_2(1)=0, \\
        g_2(\gamma)=0, \\
        g_2(\gamma^{1+q^2})=g_2(\gamma^2)=0,  \\
        g_2(\gamma^{1+q+q^2})=g_2(\gamma^3)=0.
      \end{array}
    \right.
  \end{equation*}
  Therefore, $3\leq d\leq 5$ since $g_2(x)$ is a pentanomial.

  If $m$ is odd, $s$ is odd, and there exists a proper factor $\eta$ of $q-1$ with $\eta>4$, then $\alpha^{\frac{q^m-1}{\eta}} \in\F_q^*$. Let $\gamma_1=\alpha^{\frac{q^m-1}{\eta}}$. Let
\begin{eqnarray*}
  g_3(x): &=& (x+1)(x+\gamma_1)(x+\gamma_1^2)(x+\gamma_1^3) \\
  ~ &=& x^4+\sigma_1(1,\gamma_1,\gamma_1^2,\gamma_1^3)x^3
  +\sigma_2(1,\gamma_1,\gamma_1^2,\gamma_1^3)x^2+\sigma_3(1,\gamma_1,\gamma_1^2,\gamma_1^3)x\\
  ~&~&
  +\sigma_4(1,\gamma_1,\gamma_1^2,\gamma_1^3).
\end{eqnarray*}
$\sigma_i(1,\gamma_1,\gamma_1^2,\gamma_1^3)\neq 0$ for all $i\in\{1,2,3,4\}$ since $\eta>4$.
Then $3\leq d\leq 5$ follows from the facts that $g_3(1)=g_3(\gamma_1)=g_3(\gamma_1^{1+q^2})
=g_3(\gamma_1^{1+q+q^2})=0$ and $g_3(x)$ is a pentanomial.
  \item[(ii)]$p=3$

  In this case, the parity-check matrix of $\C_s$ is
\begin{equation*}
  \left(
    \begin{array}{cccc}
      (\alpha)^{q+1} & (\alpha^2)^{q+1} & \cdots & (\alpha^{n-1})^{q+1} \\
      \alpha^{1+q^2} & (\alpha^2)^{1+q^2} & \cdots & (\alpha^{n-1})^{1+q^2} \\
      \alpha^{1+q+q^2} & (\alpha^2)^{1+q+q^2} & \cdots & (\alpha^{n-1})^{1+q+q^2}
    \end{array}
  \right)
\end{equation*}
if $\mathbb{N}_p(m)=0$ and is
\begin{equation*}
  \left(
    \begin{array}{cccc}
      1 & 1 & \cdots & 1 \\
      (\alpha)^{q+1} & (\alpha^2)^{q+1} & \cdots & (\alpha^{n-1})^{q+1} \\
      \alpha^{1+q^2} & (\alpha^2)^{1+q^2} & \cdots & (\alpha^{n-1})^{1+q^2} \\
      \alpha^{1+q+q^2} & (\alpha^2)^{1+q+q^2} & \cdots & (\alpha^{n-1})^{1+q+q^2}
    \end{array}
  \right)
\end{equation*}
if $\mathbb{N}_p(m)=1$.
From Lemma \ref{lem29}, we know $d\geq 3$. If there exists an integer $\eta_1>3$ satisfying $\eta_1|q-1$ and $\eta_1\neq q-1$.
Let $\gamma_2=\alpha^{\frac{q^m-1}{\eta_1}}$. Then $\gamma_2\in\F_q^*$,  $\gamma_2^{q+1}=\gamma_2^{1+q^2}=\gamma_2^2$,
and $\gamma_2^{1+q+q^2}=\gamma_2^3$. Let
\begin{eqnarray*}
  g_4(x): &=& (x-\gamma_2^2)(x-\gamma_2^{3}) \\
  ~ &=&  x^2-\sigma_1(\gamma_2^2,\gamma_2^3)x
+\sigma_2(\gamma_2^2,\gamma_2^3)
\end{eqnarray*}
and let
\begin{eqnarray*}
  g_5(x): &=& (x-1)(x-\gamma_2^2)(x-\gamma_2^{3}) \\
  ~ &=&  x^3-\sigma_1(1,\gamma_2^2,\gamma_2^3)x^2
+\sigma_2(1,\gamma_2^2,\gamma_2^3)x-
  \sigma_3(1,\gamma_2^2,\gamma_2^3).
\end{eqnarray*}

If $\mathbb{N}_p(m)=0$, from the facts that
$g_4(\gamma_2^2)
=g_4(\gamma_2^3)=0$ and $\sigma_{i}(\gamma_2^2,\gamma_2^3)\neq0$ for $i\in\{1,2\}$, then $d=3$.
If $\mathbb{N}_p(m)=1$, from the facts that
$g_5(1)=g_5(\gamma_2^2)
=g_5(\gamma_2^3)=0$ and $\sigma_{i}(1,\gamma_2^2,\gamma_2^3)\neq0$ for $i\in\{1,2,3\}$, then $3\leq d\leq 4$.
  \item[(iii)]$p>3$

  In this case, the parity-check matrix of $\C_s$ is
\begin{equation*}
  \left(
    \begin{array}{cccc}
    \alpha & \alpha^2 & \cdots & \alpha^{n-1} \\
      (\alpha)^{q+1} & (\alpha^2)^{q+1} & \cdots & (\alpha^{n-1})^{q+1} \\
      \alpha^{1+q^2} & (\alpha^2)^{1+q^2} & \cdots & (\alpha^{n-1})^{1+q^2} \\
      \alpha^{1+q+q^2} & (\alpha^2)^{1+q+q^2} & \cdots & (\alpha^{n-1})^{1+q+q^2}
    \end{array}
  \right)
\end{equation*}
if $\mathbb{N}_p(m)=0$ and is
\begin{equation*}
  \left(
    \begin{array}{cccc}
      1 & 1 & \cdots & 1 \\
      \alpha & \alpha^2 & \cdots & \alpha^{n-1} \\
      (\alpha)^{q+1} & (\alpha^2)^{q+1} & \cdots & (\alpha^{n-1})^{q+1} \\
      \alpha^{1+q^2} & (\alpha^2)^{1+q^2} & \cdots & (\alpha^{n-1})^{1+q^2} \\
      \alpha^{1+q+q^2} & (\alpha^2)^{1+q+q^2} & \cdots & (\alpha^{n-1})^{1+q+q^2}
    \end{array}
  \right)
\end{equation*}
if $\mathbb{N}_p(m)=1$.
\begin{itemize}
  \item If $p=5$, then $4|q-1$ and $\alpha^{\frac{q^m-1}{4}}\in\F_q^*$. Let $\gamma_3=\alpha^{\frac{q^m-1}{4}}$ and let $g_6(x)=(x-\gamma_3)(x-\gamma_3^2)(x-\gamma_3^3)$.
  It is easy to check that
   $g_6(x)\in \F_q[x]$ and $g_6(x)$ is a tetranomial. Therefore, $3\leq d\leq 4$ if $\mathbb{N}_p(m)=0$. Note that $(x-1)g_6(x)$ is a binomial instead of a pentanomial. Therefore, when $\mathbb{N}_p(m)=1$, we still need to discuss case by case. We shall omit the proofs which can be  performed similarly.
  \item If $p>5$, then $p-1>4$. Let $\gamma_4=\alpha^{\frac{q^m-1}{p-1}}$ and let $$g_7(x)=(x-1)(x-\gamma_4)
      (x-\gamma_4^2)(x-\gamma_4^{3}).$$ By imitating the above-proof steps, we can get the desired conclusion.
\end{itemize}
\end{enumerate}
This completes the proof.
\end{proof}
Below we present some explicit computational examples.
\begin{example}
\begin{itemize}
\item Let $(m,h,q)=(6,3,2)$ and $\alpha$ be a generator of $\F_{2^m}$ with $\alpha^6 + \alpha^4 + \alpha^3 + \alpha + 1=0$. Then the generator polynomial of the  binary code $\C_s$ is $\mathbb{M}_s(x)=x^{18} + x^{16} + x^8 + x^7 + x^5 + x^2 + 1
$ and $\C_s$ is a $[63,45,3]$ binary cyclic code.
\item Let $(m,h,q)=(6,3,4)$ and $\alpha$ be a generator of $\F_{4^m}$ with $\alpha^6 + \alpha^5 + w\alpha^4 + w^2\alpha^3 + \alpha^2 + \alpha + w=0$, where $w$ is a primitive element of $\F_4$. Then the generator polynomial of the code $\C_s$ is $\mathbb{M}_s(x)=x^{18} + w^2x^{17} + wx^{16} + wx^{15} + x^{12} + wx^{11} + w^2x^{10} + w^2x^9 + w^2x^6 + x^5 + wx^4 + w^2x + 1
$ and $\C_s$ is a $[4095,4077,4]$ cyclic code over $\F_4$.
\item Let $(m,h,q)=(7,3,4)$ and $\alpha$ be a generator of $\F_{4^m}$ with $\alpha^7 + \alpha^4 + \alpha^2 + w\alpha + w=0$, where $w$ is a primitive element of $\F_4$. Then the generator polynomial of the code $\C_s$ is $\mathbb{M}_s(x)=x^{22} + w^2x^{21} + w^2x^{19} + x^{18} + x^{16} + w^2x^{13} + wx^{12} + wx^{11} + w^2x^{10} + x^8 + wx^7 + x^6 + x^5 + wx^3 + x +1
$ and $\C_s$ is a $[16383,16361,d]$ cyclic code over $\F_4$.
\end{itemize}
\end{example}
\begin{remark}
When $m>6$ and $p>2$, the parameter of $\C_s$ is very large. Considering the huge amount of computation, it is difficult for us to use a Magma program to verify the minimum distance of $\C_s$.
\end{remark}
Inspired by Theorem \ref{thm33}, we give a upper bound for the minimum distance of the cyclic code $\C_s$ in
\cite[Corollary 5.24]{DSIAM} if $m$ is even. The following result provides a partial answer to the Open problem 5.16 proposed in \cite{DSIAM}.
Meanwhile, we  provide the  correct parameters of the codes considered in  \cite[Corollary 5.24]{DSIAM}.
\begin{theorem}
Let $h=3$ and $m$ satisfy $\gcd(3,m)=1$ and $m\geq 7$.
Let $s^{\infty}$ be the sequence of {\rm{(\ref{infty})}}, where $f(x)=x^{\frac{3^h+1}{2}}$. If $m$ is even, then the ternary code $\C_s$ has parameters $[3^m-1,3^m-2-7m,d]$ and generator polynomial
$$\mathbb{M}_s(x)=
(x-1)m_{\alpha^{-1}}(x)m_{\alpha^{-2}}(x)
m_{\alpha^{-5}}(x)m_{\alpha^{-10}}(x)m_{\alpha^{-11}}(x)m_{\alpha^{-13}}(x)m_{\alpha^{-14}}(x).$$
In addition,
$5\leq d\leq 8$.
\end{theorem}
\begin{proof}
According to \cite[Lemma 5.22]{DSIAM}, $\mathbb{N}_3(m)=1$ if $h=3$.
Then we only need to prove that $d\leq 8$. Since $m$ is even, $8|3^m-1$. Let $\gamma_5=\alpha^{\frac{3^m-1}{8}}$ and let
$g_8(x):=\sum_{i=0}^7(-1)^{i+1}x^i$. It is not difficult to check that $g_8(x)$ is a factor of $x^8-1$ over $\F_3$. It then follows that
\begin{equation*}
  \left\{
     \begin{array}{ll}
       g_8(1)=1-1+1-1+1-1+1-1=0, \\
       g_8(\gamma_5)=\sum_{i=0}^7(-1)^{i+1}\gamma_5^i=0, \\
       g_8(\gamma_5^2)=2\sum_{i=0}^3(-1)^{i+1}(\gamma_5^2)^i=0, \\
       g_8(\gamma_5^5)=g_8(\gamma_5)=0, \\
       g_8(\gamma_5^{10})=g_8(\gamma_5^2)=0, \\
       g_8(\gamma_5^{11})=g_8(\gamma_5)=0, \\
       g_8(\gamma_5^{13})=g_8(\gamma_5)=0, \\
       g_8(\gamma_5^{14})=g_8(\gamma_5^2)=0.
     \end{array}
   \right.
\end{equation*}
 Since $g_8(x)$ is an eight-term polynomial, $d\leq 8$.
\end{proof}
\begin{example}
Let $(m,h,q)=(8,3,3)$ and $\alpha$ be a generator of $\F_{4^m}$ with $\alpha^8 + 2\alpha^5 + \alpha^4 + 2\alpha^2 + 2\alpha + 2=0$. Then the generator polynomial of the code $\C_s$ is $\mathbb{M}_s(x)=x^{57} + 2x^{56} + x^{55} + 2x^{53} + 2x^{52} + 2x^{51} + x^{50} + 2x^{49} + 2x^{48} + x^{47}
    + x^{41} + 2x^{40} + 2x^{39} + 2x^{36} + 2x^{35} + 2x^{33} + x^{32} + 2x^{31} + 2x^{30}
    + x^{26} + x^{25} + x^{24} + x^{22} + x^{21} + x^{20} + x^{19} + x^{17} + 2x^{15} + 2x^{14} +
    x^{13} + x^{12} + 2x^{10} + x^9 + 2x^7 + 2x^6 + 2x^4 + 2x^3 + x + 2
$ and $\C_s$ is a $[6560,
6503,d]$ cyclic code over $\F_2$.
\end{example}

\section{Conclusion and outline of the contribution}\label{Conclusion}

This paper deals with cyclic codes from functions. It is based on fascinating results from Ding and a joint work of Ding and Zhou, emphasizing the role of some cryptographic functions in designing attractive binary codes.

The main contributions of this paper are listed below.
\begin{enumerate}
 \item [(1)] We complemented some results on the cyclic codes derived from the Gold,  Kasami, and Bracken-Leander functions. We focused on these functions since they have exciting characteristics: planar function, APN function, or $4$-uniform DDT function (under some restrictions on  some settings; with the notation used in the paper: $\{q,m,h\}$).
 \item [(2)] We gave partial answers to three open problems that arose in two articles  \cite{DZ} and  \cite{DSIAM} (precisely Open problem 1 raised in \cite{DZ}; while  Open problem 5.16, and Open problem 5.25  raised in  \cite{DSIAM}).
 \item [(3)] We fixed some minor errors in \cite{DSIAM,DZ} and provided correct proofs and statements.
\end{enumerate}
For \cite[Open problem 1]{DZ}, \cite[Open problems 5,16 and 5.25]{DSIAM}, we only gave partial answers.  Readers interested in working on this topic could tackle the remaining cases.

\section*{Acknowledgement}
This research is supported by National Natural Science Foundation of China (12071001). This research is supported by China Scholarship Council. The
authors would like to thank Prof. Cunsheng Ding for helpful discussions.

\end{document}